\ifdefined\IEEEBUILD
  \let\texprimitiveyear\year
\fi
\ifdefined\EXTERNALCLASS
\else
  \ifdefined\IEEEBUILD
    \documentclass{ieeeaccess}
  \else
    \documentclass[sn-mathphys-num,pdflatex]{sn-jnl}
  \fi
\fi
\ifdefined\IEEEBUILD
  \let\tqeyearsetter\year
  \let\year\texprimitiveyear
  \vol{XX}
  \tqeyearsetter{2026}
  
\fi

\ifdefined\IEEEBUILD\else
  \usepackage{fix-cm}
\fi
\usepackage{graphicx}
\makeatletter
\@ifundefined{KV@Gin@alt}{\define@key{Gin}{alt}{}}{}
\makeatother
\usepackage{amsmath,amssymb,amsfonts}
\usepackage{amsthm}
\usepackage{booktabs}
\usepackage{mathtools}

\usepackage[table]{xcolor}
\ifdefined\IEEEBUILD
  \definecolor{accessblue}{cmyk}{1,.3,0,.2}
  \definecolor{greycolor}{cmyk}{0,0,0,.8}
  \definecolor{grey}{cmyk}{0,0,0,.1}
\fi

\usepackage{tikz}

\ifdefined\IEEEBUILD
  \usepackage{cite}
\fi
\usepackage{hyperref}
\hypersetup{hidelinks,bookmarksdepth=subsubsection,pdflang={en-US},
pdftitle={Correcting Connectivity in Arc-Based QUBO Models for Fixed-Fleet Vehicle Routing}}

\renewcommand{\le}{\leqslant}

\renewcommand{\ge}{\geqslant}

\ifdefined\IEEEBUILD
\newtheoremstyle{thmstyleone}%
  {6pt}{6pt}{\itshape}{}{\bfseries}{.}{.5em}{}
\newtheoremstyle{thmstyletwo}%
  {6pt}{6pt}{\normalfont}{}{\bfseries}{.}{.5em}{}
\newtheoremstyle{thmstylethree}%
  {6pt}{6pt}{\normalfont}{}{\bfseries}{.}{.5em}{}
\fi
\theoremstyle{thmstyleone}
\newtheorem{theorem}{Theorem}

\newtheorem{lemma}{Lemma}

\theoremstyle{thmstyletwo}

\theoremstyle{thmstylethree}

\ifdefined\IEEEBUILD
\else
\fi

\newif\iflongversion
\ifdefined\INCLUDELONGCOMPARISON
  \longversiontrue
\else
  \ifdefined\IEEEBUILD
    \longversionfalse
  \else
    \longversiontrue
  \fi
\fi
\newif\iflegacymubs
\ifdefined\INCLUDELEGACYMUBS
  \legacymubstrue
\else
  \legacymubsfalse
\fi

\newif\ifieeesubmittednotice
\ifdefined\IEEESUBMITTEDNOTICE
  \ieeesubmittednoticetrue
\else
  \ieeesubmittednoticefalse
\fi

\ifdefined\IEEEBUILD\else
\date{}
\makeatletter
\renewcommand\small{\@setfontsize\small{10bp}{12bp}}
\renewcommand\footnotesize{\@setfontsize\footnotesize{10bp}{12bp}}
\renewcommand\scriptsize{\@setfontsize\scriptsize{10bp}{12bp}}
\def\abstractheadfont{\reset@font\fontsize{10bp}{12bp}\bfseries\selectfont\titraggedcenter}
\def\abstractsubheadfont{\reset@font\fontsize{10bp}{12bp}\bfseries\selectfont}
\def\abstractfont{\reset@font\fontsize{10bp}{12bp}\selectfont\leftskip=24pt\rightskip=24pt\parfillskip=0pt plus 1fil}
\def\keywordfont{\reset@font\fontsize{10bp}{12bp}\selectfont\leftskip=24pt\rightskip=24pt plus0.5fill}
\def\figurecaptionfont{\reset@font\fontfamily{\rmdefault}\fontsize{10bp}{12bp}\selectfont}
\def\tablecaptionfont{\reset@font\fontsize{10bp}{12bp}\selectfont}
\def\tablebodyfont{\reset@font\fontsize{10bp}{12bp}\selectfont}
\def\tablecolheadfont{\reset@font\fontsize{10bp}{12bp}\selectfont\bfseries\boldmath}
\def\tablefootnotefont{\reset@font\fontsize{10bp}{12bp}\selectfont}
\def\ps@plain{\let\@mkboth\@gobbletwo\let\@oddhead\@empty\let\@evenhead\@empty\def\@oddfoot{\hfil\thepage\hfil}\let\@evenfoot\@oddfoot}
\let\ps@headings\ps@plain
\let\ps@titlepage\ps@plain
\makeatother
\fi

\begin{document}

\ifieeesubmittednotice
\begin{center}
\small This work has been submitted to the IEEE for possible publication.
Copyright may be transferred without notice, after which this version may no
longer be accessible.
\end{center}
\fi

\ifdefined\IEEEBUILD
\history{Date of publication xxxx 00, 0000, date of current version xxxx 00, 0000.}
\doi{10.1109/TQE.2020.DOI}
\title{Correcting Connectivity in Arc-Based QUBO Models for Fixed-Fleet Vehicle Routing}
\author{\uppercase{Omer Gurevich}\authorrefmark{1},
\uppercase{Maor Matityahu}\authorrefmark{1},
\uppercase{Tal Mor}\authorrefmark{1}, and
\uppercase{Aryeh Lev Zabokritskiy (Yohananov)}\authorrefmark{2}}
\address[1]{Department of Computer Science, Technion--Israel Institute of
Technology, Haifa, Israel (emails: omergu@campus.technion.ac.il;
maor.m@campus.technion.ac.il; talmo@cs.technion.ac.il)}
\address[2]{Department of Computer Science, MIGAL--Galilee Research
Institute / Tel-Hai Academic College, Kiryat Shmona 11016, Israel
(email: yuhanalev@telhai.ac.il; ORCID:
\href{https://orcid.org/0000-0003-3151-6192}{0000-0003-3151-6192})}
\markboth
{Gurevich \headeretal: Correcting Connectivity in Arc-Based QUBO Models}
{Gurevich \headeretal: Correcting Connectivity in Arc-Based QUBO Models}
\corresp{Corresponding author: Aryeh Lev Zabokritskiy (Yohananov)
(email: yuhanalev@telhai.ac.il).}
\else
\title[
Correcting Connectivity in Arc-Based QUBO Models
]{
Correcting Connectivity in Arc-Based QUBO Models for Fixed-Fleet Vehicle Routing
}
\author[1]{\fnm{Omer} \sur{Gurevich}}
\email{omergu@campus.technion.ac.il}

\author[1]{\fnm{Maor} \sur{Matityahu}}
\email{maor.m@campus.technion.ac.il}

\author[1]{\fnm{Tal} \sur{Mor}}
\email{talmo@cs.technion.ac.il}

\author[2]{\fnm{Aryeh~Lev} \sur{Zabokritskiy (Yohananov)}}
\email{yuhanalev@telhai.ac.il}

\affil[1]{%
  \orgdiv{Department of Computer Science},
  \orgname{Technion -- Israel Institute of Technology},
  \orgaddress{\city{Haifa}, \country{Israel}}
}

\affil[2]{%
  \orgdiv{Department of Computer Science},
  \orgname{MIGAL -- Galilee Research Institute / Tel-Hai Academic College},
  \orgaddress{\city{Kiryat Shmona}, \postcode{11016}, \country{Israel}}
}
\fi

\ifdefined\IEEEBUILD
\begin{abstract}
\else
\abstract{
\fi
We revisit a degree-only arc Hamiltonian for fixed-fleet, homogeneous,
uncapacitated vehicle routing. Because its local penalties define only a cycle
cover, ground states may contain customer cycles disconnected from the depot.
We construct a polynomial-size quadratic unconstrained binary optimization (QUBO)
repair using capped single-commodity flow and prove that every
ground-state routing is connected and cost-optimal under explicit penalty
assumptions. For $N-1$ customers and $K$ nonempty routes, the unreduced
encoding uses exactly
$|E|(1+\lceil\log_2(N-K+1)\rceil)$ logical problem qubits. A reversible
compute--phase--uncompute realization evaluates the flow penalties in
$O(N^2\log N+N\log^2N)$ logical gates on a complete graph with $O(\log N)$
reusable workspace and no product register.
On complete loopless graphs, a position encoding that arranges all routes in
one sequence separated by depot occurrences uses fewer problem qubits and
fewer written terms when the number of bits per flow register grows. Under
the same reversible accounting model, the flow formulation admits a smaller
upper bound on the number of logical gates.
Exact audits of the Hamiltonian and circuit implementation, combined with a
classical benchmark of $1{,}200$ cost matrices, verify the formulation and
quantify the connectivity gap.
Finally, a 32,000-shot Amazon Braket task on IQM Emerald characterizes depth-one termwise Ising circuits on a diagnostic $N = 4,\, K = 1$ counterexample instance.
In the degree-only circuit, $78.05\%$ of selected $p=1$ shots realize the
invalid disconnected ground state; the reduced 14-qubit flow-augmented circuit
yields no fully feasible sample. 
These device results characterize mapped Hamiltonians and compilation rather
than an asymptotic advantage in solving routing problems.
\ifdefined\IEEEBUILD
\end{abstract}
\else
}
\fi

\ifdefined\IEEEBUILD
\begin{keywords}
quadratic unconstrained binary optimization, quantum circuit compilation,
single-commodity flow, vehicle routing
\end{keywords}
\titlepgskip=-15pt
\maketitle
\else
\keywords{Quadratic Unconstrained Binary Optimization, Quantum Circuit
Compilation, Single-Commodity Flow, Vehicle Routing}
\pacs[MSC Classification]{68Q17, 90C27, 90B06}
\maketitle
\fi

\section{Introduction}
\label{intro}

We study the fixed-fleet,
homogeneous, uncapacitated, single-depot VRP with
exactly $K$ nonempty routes. Equivalently, it is a fixed-fleet
multiple traveling-salesperson problem in the terminology surveyed by
Bekta\c{s}~\cite{Bektas2006}. 
Because these models represent transportation and logistics plans, a
depot-disconnected ground state is an infeasible plan rather than merely a
different lower-energy solution.

Routing problems have been encoded as Ising or QUBO objectives for quantum
annealing and variational quantum algorithms. Relevant formulations and
solution strategies include the TSP encoding of Lucas~\cite{Lucas2014},
hybrid and annealing-based VRP approaches~\cite{Feld2019,Irie2019}, a
comparison of quantum routing formulations~\cite{Harwood2021}, the
arc-based model of Azad et al.~\cite{Azad2023}, and later approaches to
heterogeneous, capacitated, and time-window variants
\cite{Fitzek2024,Xie2024,Masuda2023,Leonidas2024QubitEfficient}. Closely
related resource and formulation comparisons include the GPS construction
of Gonz\'alez-Bermejo et al.~\cite{GonzalezBermejo2022}, the small-instance
encoding of Mohanty et al.~\cite{Mohanty2024}, and the decomposition study of
Palackal et al.~\cite{Palackal2023}. A recent colored-permutation formulation
for capacitated routing gives a different resource tradeoff~\cite{Onah2026}.
Cattelan and
Yarkoni~\cite{CattelanYarkoni2024} emphasize that explicit subtour exclusions
can make edge-based routing QUBOs impractical; the compact extended flow
formulation below provides a polynomial-size alternative, at the price of
additional flow bits and denser interactions.

The arc formulation of Azad et al. has been reused or reproduced in several
later works. A citation alone does not establish inheritance of a modeling
gap, so our audit compares displayed constraints and Hamiltonians rather than
using aggregate citation counts. Two formulations by Mohanty et~al.~\cite{MohantyTQE2023,Mohanty2024} cite Azad et~al.\ for the arc model and list Miller--Tucker--Zemlin (MTZ)-style inequalities. However, they do not encode these inequalities into their Hamiltonian, leaving isolated subtours unpenalized. Alsaiyari and Felemban~\cite{Alsaiyari2023} display
the same arc objective and four degree/depot equations in their
Eqs.~(1)--(5), without a connectivity condition, and then benchmark the
variational quantum eigensolver (VQE) and the quantum approximate optimization
algorithm (QAOA) on that formulation. The review of Niu et
al.~\cite{Niu2026Review} reproduces the same five-term Hamiltonian without a
connectivity term. These observations concern the published equations and
experiments; works that merely cite Azad are excluded from the classification.

The issue examined here is precise. The degree constraints in the arc-based
Hamiltonian of Azad et al.~\cite{Azad2023} require every customer to have
one selected incoming and one selected outgoing arc, but these local
equations define only a cycle-cover-type subgraph: after the depot is split
into $K$ copies, customer-only cycles may remain. They do not force every
selected cycle to contain a depot copy. We give a fully specified instance in which the optimal
degree-feasible assignment contains a depot-disconnected subtour. Classical
formulations address this issue either through subset-based
subtour-elimination inequalities~\cite{Dantzig1954} or through compact
auxiliary-flow formulations, notably the single-commodity construction of
Gavish and Graves~\cite{GavishGraves1978} and its subsequent analysis by
Gouveia~\cite{Gouveia1995}. Bounded integer-to-binary maps and exact Boolean
implication penalties are also standard QUBO tools
\cite{KarimiRonagh2019,Glover2019}.

There are several quantum-compatible responses to the same feasibility
problem. Li et al.~\cite{Li2025InfeasibleConstraints} use infeasible-solution
constraints in joint and stepwise procedures for subloop elimination. Xie et
al.~\cite{Xie2024} restrict the search through a feasibility-preserving
solver. Azfar et al.~\cite{Azfar2026} identify subtours in the reported
QAOA results of Azad et al. and encode subset-based connectivity cuts. Gromiec
et al.~\cite{Gromiec2026} recently evaluated an integer arc-load
balance-and-support formulation for the equal-demand capacitated VRP (CVRP)
through a D-Wave constrained quadratic model (CQM) workflow. These approaches
address related feasibility questions through constrained search, subset
cuts, or a hybrid CQM workflow. Our target is a single static QUBO that repairs
the published degree-only fixed-fleet arc model, with an explicit
ground-state guarantee and separate counts for its binary variables,
quadratic terms, and circuit implementation.

Our contribution is an equation-level reconstruction of the failure in the published degree-only arc Hamiltonian and a polynomial-size QUBO repair of that specific model. Our repair maps each arc flow into a compact register of $L = \lceil\log_2(N - K + 1)\rceil$ bits based on the tight fixed-fleet bound $U = N - K$, enforcing bitwise support without auxiliary slack variables. Under explicit penalty bounds, we prove ground-state correctness and establish separated resource counts for logical problem qubits ($|E|(1 + L) = O(|E|\log N)$), written Hamiltonian interactions, and reversible phase-oracle synthesis. For the latter, rather than materializing the entire static quadratic expansion of $\Theta(N^3L^2)$ Pauli interactions on a complete directed graph, we adapt program-based reversible arithmetic techniques established in the QAOA literature~\cite{Bako2025} to compute and phase the balance residuals directly using reusable ancillae.

We state the underlying gate model and separate the logical arithmetic count
from the overhead of rotation synthesis and hardware connectivity. For noisy
intermediate-scale quantum (NISQ) devices, this separation makes the choice of
formulation depend on circuit compilation: the number of binary variables,
the expanded interaction count, and the cost of the reversible implementation
need not rank the encodings in the same order.

\iflongversion
We also analyze a position encoding that arranges all routes in one cyclic
sequence. Its $K$ depot occurrences separate the $K$ routes, and a penalty on
consecutive depot occurrences ensures that every route is nonempty. The
written QUBO uses $NT$ variables, where $T=N+K-1$, and contains
$\Theta(NT^2+TN^2)$ quadratic terms. Fixing the initial depot column saves
exactly $N$ variables. When $K=1$, further elimination of the depot
variables gives the same feasible tours and objective values as the
fixed-depot TSP; it does not make the unreduced Hamiltonian polynomials
identical. The $NT$ count relies on concatenating the route blocks and is not
a formulation-independent count for vehicle-indexed position models.
\fi

The remainder of the paper is organized as follows.
Section~\ref{sec:caveat} gives the counterexample and the equation-level audit.
Section~\ref{sec:flow-repair} constructs the repair, proves its ground-state
correctness, and analyzes its resources.
\iflongversion{} Section~\ref{subsec:solving-vrp} develops the
position-indexed alternative and its forced-variable reductions.\fi{}
Section~\ref{sec:experiments} gives direct Hamiltonian verification, a
systematic exact connectivity benchmark, and a small-device characterization
of the corresponding termwise Ising circuits.
Section~\ref{sec:conclusion} summarizes the engineering trade-offs,
limitations, and open directions.

\iflongversion
\section{Preliminaries and Previous Work}
The position-indexed TSP construction recalled below supplies the building
blocks for the alternative VRP encoding in Section~\ref{subsec:solving-vrp}.
Readers interested first in the arc-model repair may proceed to the precise
VRP setting in Section~\ref{subsec:vrp} and the connectivity gap in
Section~\ref{sec:caveat}.
\subsection{Hamiltonian path and Hamiltonian cycle}
\label{subsec:hamiltonian}

Throughout this TSP discussion assume $N\ge3$. Given a directed or undirected
graph $G=(V,E)$, a Hamiltonian path visits every
vertex exactly once. A Hamiltonian cycle is a cycle that visits every vertex
exactly once, apart from repeating its initial vertex to close the cycle.

If a Hamiltonian cycle is represented by its cyclic vertex order, every cyclic
shift represents the same cycle. For an undirected graph, reversing the order
also represents the same cycle. Fixing one vertex at the first position removes
the cyclic-shift redundancy and eliminates variables that would otherwise encode
equivalent representations.

Because the encodings below use ordered transitions, we adopt a uniform arc
convention. For $V=\{1,\ldots,N\}$ let
\[
\mathcal A=\{(u,v)\in V\times V:u\ne v\}.
\]
We take $E\subseteq\mathcal A$ and write
$\overline E=\mathcal A\setminus E$. An undirected edge is represented by its
two orientations, with equal costs when weights are present. Thus all nonedge
sums below range only over distinct vertices; depot self-transitions are treated
explicitly in the VRP formulation.

Deciding whether a directed or undirected graph contains a Hamiltonian path or
cycle is NP-complete. We use the position-indexed formulation of
Lucas~\cite{Lucas2014}. For vertices $v=1,\ldots,N$, let
$x_{v,t}=1$ if and only if vertex $v$ occupies position $t$ in the cyclic
order.

In a valid solution, according to~\cite{Lucas2014}, the following constraints must hold:
\begin{align}
\label{eq:con1Lucas}
    & \sum_{t=1}^{N} x_{v,t}=1, \quad v=1,2,\ldots,N,\\
\label{eq:con2Lucas}
    & \sum_{v=1}^{N} x_{v,t}=1, \quad t=1,2,\ldots,N,\\
\label{eq:con3Lucas}
    & \sum_{(u,v)\in\overline E} x_{u,t} x_{v,t+1} = 0, \quad t=1,2,\ldots,N.
\end{align}

Equation~\eqref{eq:con1Lucas} places every vertex exactly once,
Eq.~\eqref{eq:con2Lucas} places exactly one vertex at each position, and
Eq.~\eqref{eq:con3Lucas} forbids transitions that are not arcs of the graph.
Together they encode a Hamiltonian cycle as a zero of
\begin{align}
\label{eq:HCyc}
    H_A
    &=A\sum_{v=1}^{N} \left(1-\sum_{t=1}^{N} x_{v,t}\right)^2
    \notag\\
    &\quad+A\sum_{t=1}^{N} \left(1-\sum_{v=1}^{N} x_{v,t}\right)^2
    \notag\\
    &\quad+A \sum_{(u,v)\in\overline E}\sum_{t=1}^{N}x_{u,t}x_{v,t+1},
\end{align}
for some constant $A>0$, where we identify $N+1$ with $1$. For the Hamiltonian path problem, the last sum should run up to $N-1$ instead of $N$. Any assignment of the variables $x_{v,t}$ for which $H_A=0$ encodes a valid Hamiltonian cycle.

\subsection{The Traveling Salesman Problem}
\label{subsec:tsp}

For a weighted directed or undirected graph $G=(V,E)$ with nonnegative arc
costs $c_{uv}$, the Traveling Salesman Problem (TSP) seeks a minimum-cost
Hamiltonian cycle. We therefore add the travel-cost term
\begin{align}
\label{eq:HB}
    H_B = B \sum_{(u,v)\in E} c_{uv} \sum_{t=1}^{N} x_{u,t} x_{v,t+1},
\end{align}
where $B>0$ sets the objective scale.

The total Hamiltonian encoding the TSP is then given by
\begin{align}
\label{eq:HTSP}
    H_{\mathrm{TSP}} = H_A + H_B,
\end{align}
where $H_A$ is defined in Eq.~\eqref{eq:HCyc}. 
Assume that the graph contains at least one Hamiltonian cycle. Choosing the
penalty weight $A$ sufficiently large compared with the objective scale
ensures that constraint violations cannot be compensated by decreasing the
travel-cost term. In particular, if
\[
A \;>\; N\cdot B \cdot \max_{(u,v)\in E} c_{uv},
\]
then any assignment that violates at least one of
Eqs.~\eqref{eq:con1Lucas}--\eqref{eq:con3Lucas} has energy strictly larger
than every feasible Hamiltonian cycle. Indeed, every feasible cycle has
objective value at most
$NB\max_{(u,v)\in E}c_{uv}$, whereas any nonzero integer penalty is at least
$A$ and the objective is nonnegative.

As noted in Sec.~\ref{subsec:hamiltonian}, every directed Hamiltonian
cycle has $N$ cyclic representations, while an undirected cycle has $2N$
representations when both orientations are counted. Fixing the start vertex to
$v=1$ removes the cyclic-shift redundancy, leaving one representation per
orientation and enabling a direct variable reduction.

Fixing $v=1$ as the initial position also identifies it as the depot used in
the VRP model of Sec.~\ref{subsec:vrp}. Adding $x_{1,1}=1$ gives
\begin{equation}
\label{eq:HTSP_fixed}
    \begin{split}
    H_{\mathrm{TSP}}
    &= A(1-x_{1,1})^2
    + A\sum_{v=1}^{N} \left(1-\sum_{t=1}^{N} x_{v,t}\right)^2\\
    &\quad + A\sum_{t=1}^{N} \left(1-\sum_{v=1}^{N} x_{v,t}\right)^2\\
    &\quad + A \sum_{(u,v)\in\overline E}\sum_{t=1}^{N}
    x_{u,t}x_{v,t+1}\\
    &\quad + B \sum_{(u,v)\in E} c_{uv}\sum_{t=1}^{N}
    x_{u,t}x_{v,t+1}.
    \end{split}
\end{equation}

The constraint fixes $x_{1,1}$ and, together with the one-hot equations, forces
the other variables in its row and column.
\iflongversion
The forced-variable reductions in Sec.~\ref{sec:savingQbits} eliminate these
variables explicitly.
\else
They may therefore be removed before constructing the logical problem
register.
\fi

The corresponding diagonal Ising Hamiltonian follows from
$x_{v,t}\mapsto(1-\sigma^z_{v,t})/2$; a detailed TSP derivation appears in
Ref.~\cite{Gurevich2025TSP}, building on Lucas~\cite{Lucas2014}.
\fi

\iflongversion
\subsection{Vehicle Routing Problem}
\else
\section{Problem Setting}
\fi
\label{subsec:vrp}

We use the following loopless directed arc convention for the VRP. Let
$V=\{1,\ldots,N\}$ with $N\ge2$, and let $G=(V,E)$ be loopless and
directed, with $E\subseteq\{(u,v)\in V\times V:u\ne v\}$.
Node $1$ is the depot and nodes $2,\ldots,N$ are customers. Each arc
$(u,v)\in E$ has a nonnegative cost $c_{uv}\ge0$. An undirected instance is
represented by including both orientations of every undirected edge, with equal
costs.

We fix $K$ with $1\le K\le N-1$. A feasible solution is a collection of
exactly $K$ nonempty depot-to-depot routes whose customer sets are pairwise
disjoint and cover $\{2,\ldots,N\}$. Whenever a ground-state guarantee or a
penalty bound is asserted, we assume that at least one such solution exists.
This is the fixed-fleet, homogeneous, uncapacitated VRP, equivalently a
fixed-fleet multiple-TSP. Capacities, heterogeneous vehicles, physical arrival
times, and time windows are outside the present model. For $K=1$, its feasible
routes coincide with the fixed-depot TSP tours.

\section{Degree-Only Arc Model and Connectivity Gap}
\label{sec:caveat}

We revisit the arc-based Ising/QUBO formulation for the VRP proposed by
Azad et al.~\cite{Azad2023}.
While their Hamiltonian correctly enforces the local degree constraints,
it does not impose global connectivity.
As a consequence, the feasible set may contain disconnected subtours that
satisfy all encoded constraints yet do not correspond to valid vehicle routes.
This section gives an explicit counterexample, identifies the same phenomenon
in a reported numerical output, and classifies verified downstream reuse of the
published equations. Section~\ref{sec:flow-repair} then gives a compact
single-commodity-flow repair.

\subsection{Arc-Based Formulation of VRP}
\label{subsec:AzadEdgeFormulation}

Azad et al.~\cite{Azad2023} propose an arc-based formulation of the Vehicle
Routing Problem (VRP), in which routes are represented by selecting arcs of a
complete loopless directed graph.
Let node $1$ denote the depot and nodes $2,\ldots,N$ denote customers.
For each directed arc $(i,j)\in E$, a binary decision variable
$x_{ij}\in\{0,1\}$ is introduced, where $x_{ij}=1$ indicates that the route
traverses from node $i$ to node $j$.

The objective is to minimize the total travel cost,
\begin{align}
\label{eq:Azad_VRP}
    C_{\mathrm{tot}}
    = \sum_{(i,j)\in E} c_{ij}\,x_{ij},
\end{align}
where $c_{ij}$ denotes the cost of traversing arc $(i,j)$.

To ensure that each customer is visited exactly once, the formulation enforces
degree constraints on every node.
Specifically, each customer node must have exactly one outgoing and one incoming
selected arc, while the depot must have exactly $K$ outgoing and $K$ incoming
arcs, corresponding to the $K$ vehicles:
\begin{align}
\label{eq:Azad_VRP_constraints}
    &\sum_{j:(i,j)\in E} x_{ij} = 1
    \qquad &&\forall i\in\{2,\ldots,N\}, \\
    &\sum_{j:(j,i)\in E} x_{ji} = 1
    \qquad &&\forall i\in\{2,\ldots,N\}, \\
    &\sum_{j:(1,j)\in E} x_{1j} = K, \\
    &\sum_{j:(j,1)\in E} x_{j1} = K .
\end{align}
These constraints guarantee that every customer is entered and exited exactly
once and that exactly $K$ selected arcs leave and enter the depot. They do not,
by themselves, guarantee that the selected arcs form $K$ depot-rooted routes.

The objective and degree-penalty parts are
\begin{equation}
\label{eq:Azad_HA}
H_{\mathrm{obj}}
 =B\sum_{(i,j)\in E}c_{ij}x_{ij}.
\end{equation}
\ifdefined\IEEEBUILD
\begin{equation}
\label{eq:Azad_Hdegree}
\begin{aligned}
H_{\mathrm{degree}}
&=A\sum_{i=2}^{N}
\left(1-\sum_{j:(i,j)\in E}x_{ij}\right)^2\\
&\quad+A\sum_{i=2}^{N}
\left(1-\sum_{j:(j,i)\in E}x_{ji}\right)^2\\
&\quad+A\left(K-\sum_{j:(1,j)\in E}x_{1j}\right)^2\\
&\quad+A\left(K-\sum_{j:(j,1)\in E}x_{j1}\right)^2.
\end{aligned}
\end{equation}
\else
\begin{equation}
\label{eq:Azad_Hdegree}
\begin{aligned}
H_{\mathrm{degree}}
&=A\sum_{i=2}^{N}
\left(1-\sum_{j:(i,j)\in E}x_{ij}\right)^2
+A\sum_{i=2}^{N}
\left(1-\sum_{j:(j,i)\in E}x_{ji}\right)^2\\
&\quad+A\left(K-\sum_{j:(1,j)\in E}x_{1j}\right)^2
+A\left(K-\sum_{j:(j,1)\in E}x_{j1}\right)^2.
\end{aligned}
\end{equation}
\fi
Thus the degree-only arc Hamiltonian is
\begin{equation}
\label{eq:Azad_HVRP}
H_{\mathrm{arc}}=H_{\mathrm{obj}}+H_{\mathrm{degree}}.
\end{equation}
Here $A,B>0$. The objective minimizes the selected-arc cost, whereas the four
squared penalties enforce the local degree equations.

While these terms correctly enforce the local degree constraints at each
node, they do not impose any global connectivity requirement.
As a result, the formulation admits disconnected solutions composed of multiple
cycles that collectively satisfy all constraints but do not correspond to valid
vehicle routes.
An explicit counterexample illustrating this failure is shown in
Fig.~\ref{fig:disconnected}. Take $N=6$, $K=1$, and the complete loopless
directed graph. Let
\[
F=\{(1,2),(2,3),(3,1),(4,5),(5,6),(6,4)\},
\]
assign cost $1$ to the arcs in $F$, and assign cost $10$ to every other
arc. Every degree-feasible assignment selects six arcs and therefore has cost
at least $6$. The arcs in $F$ attain this bound, so they form an optimal
degree-feasible assignment, but they decompose into the disconnected cycles
$1\to2\to3\to1$ and $4\to5\to6\to4$. Any Hamiltonian cycle must cross the
partition $\{1,2,3\}\cup\{4,5,6\}$ in both directions and consequently costs
at least $2\cdot10+4\cdot1=24$.

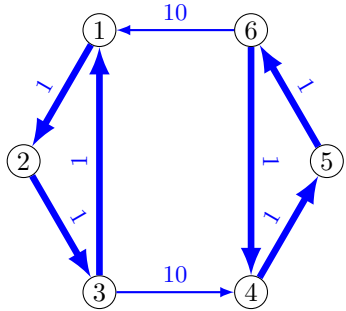
\begin{figure}[!b]
    \centering
    \begin{tikzpicture}[scale=1]

        \node[draw,circle,inner sep=1.5pt] (v1) at (-1,  1.732) {1};
        \node[draw,circle,inner sep=1.5pt] (v2) at (-2,  0)     {2};
        \node[draw,circle,inner sep=1.5pt] (v3) at (-1, -1.732) {3};
        \node[draw,circle,inner sep=1.5pt] (v4) at ( 1, -1.732) {4};
        \node[draw,circle,inner sep=1.5pt] (v5) at ( 2,  0)     {5};
        \node[draw,circle,inner sep=1.5pt] (v6) at ( 1,  1.732) {6};

        \tikzset{ed/.style={->,>=latex}}

        \draw[ed, line width=2.5pt, blue] (v1) -- node[above,sloped] {\scriptsize $1$} (v2);
        \draw[ed, line width=2.5pt, blue] (v2) -- node[above,sloped] {\scriptsize $1$} (v3);
        \draw[ed, line width=2.5pt, blue] (v3) -- node[above,sloped] {\scriptsize $1$} (v1);

        \draw[ed, line width=2.5pt, blue] (v4) -- node[above,sloped] {\scriptsize $1$} (v5);
        \draw[ed, line width=2.5pt, blue] (v5) -- node[above,sloped] {\scriptsize $1$} (v6);
        \draw[ed, line width=2.5pt, blue] (v6) -- node[above,sloped] {\scriptsize $1$} (v4);

        \draw[ed, line width=0.8pt, blue] (v3) -- node[above,sloped] {\scriptsize $10$} (v4);
        \draw[ed, line width=0.8pt, blue] (v6) -- node[above,sloped] {\scriptsize $10$} (v1);

    \end{tikzpicture}
    \caption{A fully specified counterexample for $K=1$. The six bold arcs
    have cost $1$ and form an optimal degree-feasible assignment of cost $6$,
    but they consist of two disconnected cycles. Every unshown arc has cost
    $10$; the two thin arcs illustrate a connected Hamiltonian cycle of cost
    $24$.}
    \label{fig:disconnected}
\end{figure}
Connectivity problems, commonly referred to as \emph{subtours} (disconnected
cycles), arise when an arc-based formulation enforces only local degree
constraints but does not impose a global connectivity requirement. In such
cases, the selected arcs may decompose into multiple disjoint cycles, some of
which may be disconnected from the depot. These configurations are invalid VRP
solutions, yet they can satisfy all degree constraints and hence appear as
low-energy states of a degree-only Hamiltonian. The repair used here adds
a compact single-commodity flow layer whose conservation constraints enforce
that every customer lies in the depot-connected component.

The example isolates the issue addressed in the remainder of the paper: local
degree equations select a cycle-cover-type subgraph, not necessarily a
collection of depot-rooted routes. After examining the published equations and
their reuse below, Section~\ref{sec:flow-repair} adds compact auxiliary-flow
conditions and translates them into binary penalty terms.

\subsection{Evidence in the Source Model and Downstream Reuse}
\label{subsec:downstream-audit}

A reported numerical output of Azad et al.~\cite{Azad2023} already exhibits
the omitted-connectivity phenomenon. In their Experiment~2, the adjacency
matrix $A_2$ displayed in their Eq.~(34) contains the customer-only cycle
$2\to3\to2$, disconnected from depot~$0$, in addition to the two depot routes
$0\to1\to0$ and $0\to4\to0$. It satisfies their stated customer- and
depot-degree equations and is reported with cost $128.545$, below the cost
$138.511$ of the connected matrix $A_3$ in their Eq.~(35). Thus the reported
example itself demonstrates why the degree-only objective can prefer a cheaper
disconnected degree-feasible subgraph to a feasible VRP solution.

A citation alone does not show that a later model inherits an earlier modeling
gap. We therefore compare the constraints that the selected follow-on works
state with the terms that their displayed Hamiltonians actually minimize. In
the TQE formulation of Mohanty et al.~\cite{MohantyTQE2023}, Eq.~(17) lists
MTZ-style inequalities and identifies them as subtour-elimination conditions,
whereas Eq.~(18) contains only the objective and the four degree/depot
penalties inherited from Azad et al. The later QSVM formulation
\cite{Mohanty2024} lists MTZ-style conditions in Eq.~(23), but its five-term
Hamiltonian in Eq.~(24) again contains no corresponding term; Eq.~(25) merely
defines the binary arc vector. Thus the connectivity conditions are discussed,
but are not part of either stated unconstrained Hamiltonian.

Alsaiyari and Felemban~\cite{Alsaiyari2023} display the same arc objective and
four customer/depot degree equations in their Eqs.~(1)--(5), again without a
flow, ordering, cut, or other connectivity condition. They then report
classical, VQE, and QAOA experiments for instances $(3,2)$, $(4,2)$, and
$(5,2)$. This establishes use of the incomplete formulation in their
experimental pipeline, but it does not establish that any particular
reported output is disconnected, because their full instances and code were
not reconstructed here.

Niu et al.~\cite{Niu2026Review} provide a different kind of propagation. Their
review reproduces the five-term Azad--Mohanty Hamiltonian as a VRP formulation,
again without a connectivity term, and does not
flag the disconnected-cycle possibility. Because this is a review, the finding
concerns presentation of the model rather than the correctness of a new solver.

\iflongversion
Jaroszczuk~\cite{Jaroszczuk2025} similarly describes Azad's small numerical
tests as showing model correctness, but reproduces no Azad equation and
reports no new implementation based on that formulation. We therefore
classify this as narrative propagation rather than inherited solver behavior.
\fi

These classifications concern the published equations, experiments, and
descriptions. The counterexample above supplies
the mathematical reason that the omitted condition matters.

The object audited here is the stand-alone Hamiltonian that is displayed and
minimized in the cited works. One may instead place a degree-only Hamiltonian
inside a lazy-constraint loop that detects disconnected samples and adds new
diagonal penalties iteratively, as in iterative constraint-generation
approaches~\cite{Li2025InfeasibleConstraints}. Such a procedure is a different algorithm: it
requires a separation rule, an update and termination specification, and its
own resource analysis. It does not make the original static Hamiltonian
connectivity correct.

\section{Flow-Augmented QUBO Repair}
\label{sec:flow-repair}

\subsection{Single-Commodity-Flow Conditions}
\label{subsec:flow-constraints}
A classical and well-established way to prevent disconnected solutions in
arc-based formulations of the Vehicle Routing Problem is to introduce explicit
flow constraints that enforce global connectivity. In particular,
\emph{single-commodity flow} formulations ensure that every customer is
connected to the depot through the selected arcs
\cite{GavishGraves1978,Gouveia1995}.

A recent quantum-oriented use of the same classical principle appears in
Gromiec et al.~\cite{Gromiec2026}. Their F1 formulation uses integer arc-load
variables, customer balance, and arc support as one of four equal-demand CVRP
models evaluated through D-Wave's CQM workflow. It does not give an exact
bounded binary expansion, the fixed-fleet bound $U=N-K$, bitwise QUBO support
coupling, a ground-state theorem, or the resource accounting developed below.

In addition to the binary routing variables $x_{ij}\in\{0,1\}$, we introduce
flow variables $f_{ij}\ge 0$, representing the amount of flow sent along arc
$(i,j)$. The flow originates at the depot (node $1$) and is consumed at the
customers.

A standard single-commodity flow formulation imposes the constraints below.
We use the sharper uniform capacity
\begin{equation}
\label{eq:flow-upper-bound}
U=N-K.
\end{equation}
Indeed, if the $K$ nonempty routes contain $m_1,\ldots,m_K$ customers, then
$m_r\ge1$, $\sum_r m_r=N-1$, and therefore $m_r\le N-K$ for every route.
The same bound appears explicitly as
$y_{ij}\le (|V|-K)x_{ij}$ in a classical adapted Gavish--Graves mTSP
formulation~\cite{Campuzano2020}; here it is carried into the capped binary
register. The usual bound $N-1$ is valid but needlessly loose when $K$ is
large.
\begin{align}
\label{eq:flow1}
    &\sum_{j:(j,i)\in E} f_{ji}
    -\sum_{j:(i,j)\in E} f_{ij}=1,
    \quad i = 2,\ldots,N, \displaybreak[2]\\
\label{eq:flow2}
    &\sum_{j:(1,j)\in E} f_{1j}
    -\sum_{j:(j,1)\in E} f_{j1}=N-1, \\
\label{eq:flow3}
    &0 \le f_{ij} \le U\,x_{ij},
    \quad (i,j)\in E.
\end{align}
Constraint~\eqref{eq:flow1} enforces flow conservation with unit consumption at
each customer, \eqref{eq:flow2} supplies the total demand from the depot, and
\eqref{eq:flow3} couples flow to routing decisions so that flow can traverse
only selected arcs. Since we enforce conservation only on customer nodes, the
depot balance \eqref{eq:flow2} follows automatically by summing
\eqref{eq:flow1} over all customers.
The reachability argument requires only nonnegative real flows. For the QUBO
encoding we restrict to nonnegative integer flows; this loses no feasible
routing, because every feasible route collection admits the integral flow
constructed in the proof of Theorem~\ref{thm:flow-correctness}.

Here we take unit demand at every customer, i.e., $q_i\equiv 1$, which suffices
to eliminate disconnected subtours while keeping the flow layer compact. More
general demands $q_i$ (as in, e.g.,~\cite{Fitzek2024}) can be incorporated in the
same framework at the cost of additional bookkeeping in the flow variables.

The flow-augmented Hamiltonian consists of the objective, the degree penalties,
the flow-balance penalties, and the arc--flow coupling penalty. Flow
conservation at each customer node is enforced by
\begin{equation}
\label{eq:H-flow-customer}
H_{\mathrm{flow}}
=
A_f
\sum_{i=2}^{N}
\left(
\sum_{j:(j,i)\in E} f_{ji}
-
\sum_{j:(i,j)\in E} f_{ij}
-
1
\right)^2.
\end{equation}

Summing all customer-balance residuals gives the depot-balance residual.
Consequently, the efficient Hamiltonian needs no separate squared depot term:
zero customer-balance penalty already enforces Eq.~\eqref{eq:flow2}. Omitting
that redundant square preserves the routing projection and the zero-penalty
set while reducing the written interaction count.

It remains to enforce the coupling condition~\eqref{eq:flow3}
within a quadratic QUBO/Ising Hamiltonian.
At the classical level, \eqref{eq:flow3} is a one-sided capacity bound that
forces $f_{ij}=0$ whenever $x_{ij}=0$ and upper-bounds $f_{ij}$ by $U$ when
$x_{ij}=1$.
A convenient way to express this coupling as a penalty in a \emph{classical}
(non-binary) formulation is via the positive-part operator:
\begin{equation}
\label{eq:Hcap-classical}
H_{\mathrm{cap}}^{(+)}
=
A_c
\sum_{(i,j)\in E}
\Bigl(f_{ij}-Ux_{ij}\Bigr)_+^2 ,
\end{equation}
where $(u)_+ \triangleq \max\{u,0\}$ and $A_c>0$ is a penalty weight.
Indeed, $H_{\mathrm{cap}}^{(+)}=0$ holds if and only if
$f_{ij}\le Ux_{ij}$ for all $(i,j)$, and together with the implicit
nonnegativity $f_{ij}\ge 0$ this reproduces~\eqref{eq:flow3}.

The positive-part expression in~\eqref{eq:Hcap-classical} is not a quadratic
polynomial in the existing variables $x_{ij}$ and $f_{ij}$. After an integer
binary discretization it can be encoded exactly in QUBO form by introducing
suitable slack or auxiliary bits. Below we instead use a bounded binary flow
register together with a bitwise support penalty, which gives the required
capacity coupling without additional slack registers.

\begin{lemma}[Flow support implies depot reachability]
\label{lem:flow-eliminates-subtours}
Let $x_{ij}\in\{0,1\}$ and $f_{ij}\ge0$ satisfy
\[
\sum_{j:(j,i)\in E}f_{ji}
-\sum_{j:(i,j)\in E}f_{ij}=1,
\qquad i=2,\ldots,N,
\]
and suppose that $f_{ij}>0$ implies $x_{ij}=1$ for every $(i,j)\in E$.
Then every customer is reachable from the depot in the selected directed graph
\[
G_x=\bigl(V,\{(i,j)\in E:x_{ij}=1\}\bigr).
\]
In particular, no customer lies in a depot-disconnected selected component.
\end{lemma}

\begin{proof}
Suppose that some customer is not reachable from the depot, and let
$S\subseteq\{2,\ldots,N\}$ be the nonempty set of all unreachable customers.
No selected arc enters $S$ from $V\setminus S$; otherwise its head would be
reachable. The support condition therefore makes every flow on an arc entering
$S$ equal to zero. Summing the balance equations over $i\in S$ and cancelling
the flows on arcs internal to $S$ gives
\[
\sum_{\substack{(j,i)\in E\\ j\notin S,\ i\in S}} f_{ji}
-
\sum_{\substack{(i,j)\in E\\ i\in S,\ j\notin S}} f_{ij}
=|S|.
\]
The first sum is zero and the second is nonnegative, so the left-hand side is
nonpositive, contradicting $|S|>0$. Hence every customer is reachable from the
depot.
\end{proof}

In the next subsection, we instead replace~\eqref{eq:Hcap-classical} by a purely
quadratic coupling between the binary routing variables and the binary flow
bits, which enforces that flow can be nonzero only on selected arcs while
maintaining a valid QUBO/Ising encoding.

\subsection{Bounded QUBO, Correctness, and Ising Map}
\label{subsec:flow-ising}

To obtain a valid QUBO/Ising encoding, we specialize the range-exact binary
branch of Karimi and Ronagh's bounded integer construction
\cite{KarimiRonagh2019} to represent exactly $0,\ldots,U$. Fitzek et
al.~\cite{Fitzek2024} use the same capped-bit pattern for a per-vehicle
capacity register in a position-indexed HVRP model; here it becomes a per-arc
flow register in the single-commodity-flow repair. Let
\begin{equation}
\label{eq:flow-binary-expansion}
\begin{aligned}
L&=\lceil\log_2(U+1)\rceil,\\
w_\ell&=
\begin{cases}
2^\ell, & 0\le \ell\le L-2,\\
U+1-2^{L-1}, & \ell=L-1,
\end{cases}\\
f_{ij}&=\sum_{\ell=0}^{L-1}w_\ell y_{ij,\ell}.
\end{aligned}
\end{equation}
where $y_{ij,\ell}\in\{0,1\}$. The lower $L-1$ bits represent
$0,\ldots,2^{L-1}-1$ and the final bit shifts that interval by
$U+1-2^{L-1}$. Their union covers every integer from $0$ through $U$,
possibly with duplicate representations, while never exceeding $U$.

\medskip
\noindent
While the binary encoding controls the \emph{magnitude} of $f_{ij}$, it does not
by itself ensure that flow is carried only on selected routing arcs.
To retain the single-commodity flow logic (and hence penalize disconnected
subtours), we enforce the coupling condition that flow may appear only on
selected arcs. A convenient quadratic penalty is obtained by imposing the
constraint bitwise. The resulting monomial is the standard QUBO penalty for
the Boolean implication $y_{ij,\ell}\le x_{ij}$~\cite{Glover2019}:
\begin{equation}
\label{eq:H-cap-bitwise}
H_{\mathrm{cap}}
=
A_c
\sum_{(i,j)\in E}
\sum_{\ell=0}^{L-1}
y_{ij,\ell}\,\bigl(1-x_{ij}\bigr),
\end{equation}
with $A_c>0$.
This term enforces that $x_{ij}=0$ forces every $y_{ij,\ell}=0$ and hence
$f_{ij}=0$. If $x_{ij}=1$, the bounded expansion supplies
$0\le f_{ij}\le U$. Consequently, the exact classical coupling
$0\le f_{ij}\le Ux_{ij}$ is enforced.

\medskip
\noindent
As in the classical single-commodity flow formulation, we impose flow
conservation at each customer node:
\[
\sum_{j:(j,i)\in E} f_{ji}
-\sum_{j:(i,j)\in E} f_{ij}=1,
\qquad i=2,\ldots,N.
\]
Substituting \eqref{eq:flow-binary-expansion} yields the quadratic penalty
\begin{equation}
\label{eq:H-flow-cons}
\begin{aligned}
H_{\mathrm{flow}}
&=A_f\sum_{i=2}^{N}\Biggl(
\sum_{j:(j,i)\in E}\sum_{\ell=0}^{L-1}w_\ell y_{ji,\ell}\\
&\hspace{5em}
-\sum_{j:(i,j)\in E}\sum_{\ell=0}^{L-1}w_\ell y_{ij,\ell}
-1\Biggr)^2,
\end{aligned}
\end{equation}
where $A_f>0$ is a penalty weight.

\medskip
\noindent
Combining the original arc-based objective and degree constraints with the flow
penalties, we obtain the fully-binary quadratic Hamiltonian
\begin{equation}
\label{eq:H-flow-qubo}
H_{\mathrm{VRP}}^{\mathrm{flow}}
=
H_{\mathrm{obj}}
+
H_{\mathrm{degree}}
+
H_{\mathrm{flow}}
+
H_{\mathrm{cap}}.
\end{equation}
The components in~\eqref{eq:H-flow-qubo} are defined in
Eqs.~\eqref{eq:Azad_HA}, \eqref{eq:Azad_Hdegree},
\eqref{eq:H-flow-cons}, and \eqref{eq:H-cap-bitwise}.

\begin{theorem}[Correctness of the flow-augmented Hamiltonian]
\label{thm:flow-correctness}
Assume that $G$ admits at least one feasible $K$-route solution and that
$c_{ij}\ge0$ for all $(i,j)\in E$. Let
\[
c_{\max}=\max_{(i,j)\in E}c_{ij},
\qquad
P=\min\{A,A_f,A_c\}.
\]
If
\[
P>B(N-1+K)c_{\max},
\]
then the routing part of every minimizing binary assignment of
$H_{\mathrm{VRP}}^{\mathrm{flow}}$ encodes an optimal feasible collection of
$K$ nonempty depot-to-depot routes. Conversely, every optimal feasible routing
has at least one binary flow encoding that minimizes the Hamiltonian. After the
binary-to-Ising substitution, the ground eigenspace is spanned by
computational-basis states encoding such minimizing assignments.
\end{theorem}

\begin{proof}
First fix a feasible routing. For each route
\[
1 \to v_1 \to v_2 \to \cdots \to v_m \to 1,
\]
define
\[
\begin{aligned}
f_{1,v_1}&=m,\\
f_{v_k,v_{k+1}}&=m-k\quad(k=1,\ldots,m-1),\\
f_{v_m,1}&=0,
\end{aligned}
\]
and set every flow on an unselected arc to zero. With
$v_0=v_{m+1}=1$, each customer $v_k$ has net inflow
$f_{v_{k-1},v_k}-f_{v_k,v_{k+1}}=1$. The depot supplies
$\sum_r m_r=N-1$ units, and the support penalties vanish. Moreover,
$m\le U=N-K$ because the other $K-1$ nonempty routes contain at least one
customer each. Hence $0\le f_{ij}\le U$.
Equation~\eqref{eq:flow-binary-expansion} represents
every integer in this range, so every feasible routing has a binary assignment
for which all penalties vanish.

A feasible routing selects exactly $N-1+K$ arcs. Its energy is therefore at
most $B(N-1+K)c_{\max}$. Every equality residual is an integer, and every
summand in $H_{\mathrm{cap}}$ is binary. Hence any assignment with a nonzero
penalty has energy at least $P$, because the objective is nonnegative. The
assumed strict inequality shows that no such assignment can minimize the
Hamiltonian.

Every minimizer therefore satisfies the degree, customer-balance, and support
constraints exactly. The degree equations give every customer one
selected incoming and one selected outgoing arc and give the depot $K$ of
each. Lemma~\ref{lem:flow-eliminates-subtours} makes every customer reachable
from the depot, so no customer-only cycle remains. Start a walk along each of
the $K$ selected arcs leaving the depot and then follow the unique outgoing arc
at every customer. Two such walks cannot merge, because their first common
customer would then have two selected incoming arcs. A walk cannot enter a
customer-only cycle either, because the entry vertex would again have two
selected incoming arcs. Since the graph is finite, each walk must therefore
return to the depot, using one of its $K$ selected incoming arcs. The walks
decompose the selected subgraph into exactly $K$ customer-disjoint
depot-to-depot routes. Because the graph is loopless,
each route is nonempty. On this zero-penalty set the Hamiltonian equals
$H_{\mathrm{obj}}$, so its minimizers are precisely the minimum-cost feasible
routings together with their admissible flow encodings.
\end{proof}

\noindent
Although the above flow certificate was defined by following each route in
visit order, it is a single static assignment. It certifies the existence of a
representable integer flow for every feasible connected routing. Each such
binary assignment is a computational-basis
eigenstate of the diagonal Hamiltonian; it belongs to the ground eigenspace if
and only if its routing part has minimum travel cost.

Applying the standard mapping $b\mapsto (1-\sigma_b^z)/2$ to all binary
variables in \eqref{eq:H-flow-qubo} yields the quantum Ising Hamiltonian.  For
compactness, put
$\Delta_b=1-\sigma_b^z$ and
$\mathcal H_{\mathrm q}=H_{\mathrm{VRP}}^{\mathrm{flow,quantum}}$. The exact
operator is most transparently specified by the substitution
\begin{equation}
\label{eq:H-flow-quantum}
\mathcal H_{\mathrm q}
=H_{\mathrm{VRP}}^{\mathrm{flow}}
\!\left(
x_{ij}\mapsto\frac{\Delta_{x_{ij}}}{2},\;
y_{ij,\ell}\mapsto\frac{\Delta_{y_{ij,\ell}}}{2}
\right).
\end{equation}
This substitution in Eq.~\eqref{eq:H-flow-qubo} is exact and leaves the
computational-basis energy of every binary assignment unchanged.

\medskip
\noindent
The sufficient penalty condition and its feasibility assumptions are stated in
Theorem~\ref{thm:flow-correctness}.

\subsection{Logical resources and reversible phase implementation}
\label{subsec:flow-resources}

The formulas below distinguish three resource levels.  A binary variable maps
to one \emph{logical problem qubit} in the direct Ising encoding.  Arithmetic
ancillas used to implement a phase oracle are additional logical circuit
qubits.  General physical-qubit, embedding, routing, and fault-tolerance
overheads are hardware dependent and are not part of the analytic bounds below;
Section~\ref{subsec:iqm-hardware} reports the realized mapping of one finite
termwise-Ising implementation.

The displayed arc model has one routing bit $x_{ij}$ per directed
arc and an $L$-bit flow register $(y_{ij,\ell})_{\ell=0}^{L-1}$ on that arc.
Its exact logical
problem-qubit count is therefore
\begin{equation}
\label{eq:flow-logical-qubits}
n_{\mathrm{prob}}^{\mathrm{flow}}=|E|(1+L),
\qquad L=\left\lceil\log_2(N-K+1)\right\rceil.
\end{equation}
This is the exact count for the displayed register allocation. Every feasible
routing admits a canonical certificate with
$f_{i1}=0$ on arcs entering the depot, so one may add these equalities and
eliminate the corresponding registers without losing a routing solution. The
equalities do not follow from the displayed system itself, which can also
admit bounded circulations, and the strengthening changes only the auxiliary
zero-penalty degeneracy.
For a complete loopless directed graph this becomes
\[
N(N-1)\left(1+\left\lceil\log_2(N-K+1)\right\rceil\right).
\]
It is $O(|E|\log N)$ uniformly. For fixed $K$, or more generally whenever
$K\le(1-\delta)N$ for a constant $\delta>0$, one has
$L=\Theta(\log N)$ and the count is $\Theta(|E|\log N)$. The formal
$\Theta(|E|)$ boundary occurs only in the near-saturated fleet regime
$K=N-O(1)$.

To separate algebraic size from circuit cost, let $d_i^-$ and $d_i^+$ be the
in- and out-degrees of vertex $i$.  Before repeated coefficients are
collected, the number of written quadratic-monomial occurrences in the degree,
flow-balance, and support penalties is
\begin{equation}
\label{eq:flow-occurrences}
\begin{aligned}
Q_{\mathrm{occ}}^{\mathrm{flow}}
 ={}&\sum_{i=1}^{N}\left[
\binom{d_i^+}{2}+\binom{d_i^-}{2}
\right]\\
&+\sum_{i=2}^{N}\binom{L(d_i^-+d_i^+)}{2}
+|E|L.
\end{aligned}
\end{equation}
On a complete directed graph this is
\begin{equation}
\begin{aligned}
Q_{\mathrm{occ,comp}}^{\mathrm{flow}}
={}&2N\binom{N-1}{2}
 +(N-1)\binom{2(N-1)L}{2}\\
 &+N(N-1)L
 =\Theta(N^3L^2).
\end{aligned}
\end{equation}
The same asymptotic order holds after collecting nonzero couplers, but
Eq.~\eqref{eq:flow-occurrences} deliberately counts written occurrences rather
than claiming an exact number of distinct hardware interactions.
A direct termwise phase implementation applies a phase to each collected
nonzero term and therefore has the same asymptotic gate count. The
resource-aware reversible-oracle synthesis below instead evaluates each
residual arithmetically and thereby avoids materializing that expansion. Its
compute--phase--uncompute schedule separately keeps the ancillary workspace
logarithmic.

We realize this alternative through the established programmatic
compute--phase--uncompute
paradigm~\cite{Bako2025}. Assume exact integer arithmetic, logical all-to-all
connectivity, and ripple-carry addition. We count Toffoli/CNOT-class gates for
compute--uncompute and each logical phase or controlled-phase rotation as one
operation before rotation synthesis. The sharp uniform signed-accumulator width
for the displayed loopless model is
\[
q=1+\left\lceil
\log_2\bigl((N-1)U+1\bigr)
\right\rceil=O(\log N).
\]
Indeed, during the stated schedule the accumulator lies between
$-(N-1)U-1$ and $(N-1)U$, both of which fit in this two's-complement range.
For each customer $i$, initialize a signed $q$-qubit accumulator to zero.
To add a bounded flow word, we add its lower $L-1$ ordinary binary bits and
then conditionally add the constant $w_{L-1}$ controlled by its final bit.
This implements the capped weights in Eq.~\eqref{eq:flow-binary-expansion}
even when the final weight is not a power of two, using $O(q)$ gates and
$O(q)$ reusable workspace. For $L=1$ the lower subword is empty. Subtraction
uses the inverse operations. These reversible additions add each incoming
flow word, subtract each outgoing word, and subtract the unit demand, leaving the
balance residual
\[
r_i=\sum_{j:(j,i)\in E}f_{ji}
-\sum_{j:(i,j)\in E}f_{ij}-1
\]
in the accumulator. Write its lower-order two's-complement bits as
$b_{i,0},\ldots,b_{i,q-2}$, and let $s_i$ be the sign bit. Thus
\[
r_i=\sum_{j=0}^{q-2}2^j b_{i,j}-2^{q-1}s_i.
\]
Boolean idempotence gives the exact two's-complement phase polynomial
\begin{align}
\label{eq:signed-square-phase}
r_i^2={}&
\sum_{j=0}^{q-2}2^{2j}b_{i,j}
+2^{2q-2}s_i
+\sum_{0\le j<k\le q-2}2^{j+k+1}b_{i,j}b_{i,k}
\notag\\
&-\sum_{j=0}^{q-2}2^{q+j}s_i b_{i,j}.
\end{align}
Consequently, no product register is needed. Define the logical diagonal gates
\[
P_j(\phi)\lvert z\rangle=e^{i\phi z_j}\lvert z\rangle,\qquad
CP_{jk}(\phi)\lvert z\rangle=e^{i\phi z_jz_k}\lvert z\rangle.
\]
For phase-separator angle $\gamma$, one $P$ gate per accumulator bit and one
$CP$ gate per bit pair implement
$D_i\lvert r_i\rangle=\exp(-i\gamma A_f r_i^2)\lvert r_i\rangle$ exactly on
the computational basis. Equivalently, if
$w_j=2^j$ for $j<q-1$ and $w_{q-1}=-2^{q-1}$, the angles are
$-\gamma A_f w_j^2$ and $-2\gamma A_f w_jw_k$, respectively.
Relative to explicit reversible squaring, this removes the clean $2q$-qubit
product register while retaining the $O(q^2)$ logical phase-operation count
and the same $O(q)$ asymptotic workspace order.

Writing $\lvert y\rangle$ for all persistent flow-bit registers, the complete
schedule is
\[
\begin{aligned}
U_i\lvert y\rangle\lvert0\rangle_q
&=\lvert y\rangle\lvert r_i\rangle,\\
U_i^\dagger D_iU_i\lvert y\rangle\lvert0\rangle_q
&=e^{-i\gamma A_f r_i^2}\lvert y\rangle\lvert0\rangle_q .
\end{aligned}
\]
The inverse additions therefore clear the residual accumulator, so it and any
adder carry bits can be reused for the next customer. The flow words $y$ are
persistent problem registers, not reusable ancillas. Hence the peak
logical-qubit count of this implementation is
\begin{equation}
\label{eq:flow-peak-qubits}
n_{\mathrm{peak}}^{\mathrm{flow}}
=|E|(1+L)+O(q).
\end{equation}
Each incident-word addition or subtraction costs $O(q)$ gates, and compute
and uncompute change only the constant factor. The square phase uses exactly
$q+\binom{q}{2}$ logical diagonal operations. Thus one customer balance costs
$O((d_i^-+d_i^+)q+q^2)$ elementary gates. Summing these arithmetic costs over
all customer balances gives
\begin{equation}
\label{eq:flow-gate-bound}
O\!\left(|E|q+(N-1)q^2+|E|L\right)
\end{equation}
logical gates, where the last term accounts for the $|E|L$ support phases.
For a complete graph this is
$O(N^2\log N+N\log^2N)$ with $O(\log N)$ reusable ancillas.  The degree
residuals can be evaluated arithmetically within the same asymptotic bound, and
the linear objective adds $O(|E|)$ rotations.

This is a logical-operation upper bound, not a depth, $T$-count, or
physical-resource estimate. It excludes rotation synthesis, coefficient
precision, restricted connectivity, and error correction. For example,
expanding a customer balance square creates coefficients of order $A_fU^2$;
the degree penalties and objective have their own scales, so this is not a
bound on the largest coefficient in the full Hamiltonian. Device-specific
rescaling and precision may therefore be important even when the
logical-qubit count is favorable.

\subsection{Engineering trade-off with position indexing}
\label{subsec:engineering-tradeoff}

The flow repair is not the only connectivity-correct representation. Our
single-sequence position model concatenates the $K$ route blocks into a cyclic
sequence, with $K$ depot occurrences acting as separators. With $T=N+K-1$
positions, a one-hot register $x_{v,t}$, and a no-consecutive-depot penalty,
connectivity is enforced by the ordering itself.
Table~\ref{tab:flow-position-tradeoff} compares this model with the
flow-augmented arc construction on a complete loopless graph.
For a like-for-like comparison, the reusable-workspace row applies arithmetic
compute--phase--uncompute to the squared residuals in both models; the
position model retains explicit phases for its transition table.

\begin{table}[!b]
\caption{Dense-graph trade-off for the displayed, unreduced registers. Here
$T=N+K-1$, $L=\lceil\log_2(N-K+1)\rceil$, and $q=O(\log N)$. Gate counts are
logical upper bounds per phase application, not depth or physical-resource
estimates. Under the direct termwise implementation, the gate count has the
same asymptotic order as the written-term count, so the two are combined.}
\label{tab:flow-position-tradeoff}
\centering
\footnotesize
\setlength{\tabcolsep}{1.8pt}
\begin{tabular}{@{}>{\raggedright\arraybackslash}p{.20\linewidth}
>{\raggedright\arraybackslash}p{.35\linewidth}
>{\raggedright\arraybackslash}p{.36\linewidth}@{}}
\toprule
Resource & Flow--arc & Single-sequence position\\
\midrule
Problem qubits & $N(N-1)(1+L)$ & $NT^{*}$\\
Written terms / termwise gates & $\Theta(N^3L^2)$ & $\Theta(NT^2+TN^2)$\\
Reusable-workspace gates & $O(N^2q+Nq^2)$ & $O(TN^2+NT\log N)$\\
Workspace & $O(q)$ & $O(\log N)$\\
Sparse graph & Tracks $|E|$ and node degrees & Forbidden-arc terms keep
the explicit transition table dense\\
\bottomrule
\end{tabular}
\par\smallskip
\noindent\parbox{.95\linewidth}{\raggedright\scriptsize
\({}^{*}\) The count $NT$ applies to the depot-delimited single-sequence
encoding analyzed here. A vehicle-indexed register $x_{v,t,k}$ with $S$ local
positions per vehicle uses $KNS$ variables before reductions; its qubit count
is therefore formulation- and parameter-dependent.\par}
\end{table}

For the complete graph, the displayed flow register never uses fewer problem
qubits:
\[
n_{\mathrm{prob}}^{\mathrm{flow}}-n_{\mathrm{prob}}^{\mathrm{pos}}
=N\bigl((N-1)L-K\bigr)\ge0,
\]
with equality only at $K=N-1$; eliminating the forced first position makes
the position register smaller still. This inequality is specific to the
single-sequence register. A direct vehicle-indexed one-hot model $x_{v,t,k}$
with $S$ local positions per vehicle uses $KNS$ problem variables before
reductions. It can therefore use more or fewer qubits than the flow register,
depending on $K$ and $S$, while retaining vehicle-local state that makes timing
and vehicle-specific constraints more direct. When $L$ grows, the
single-sequence model's $\Theta(N^3)$
explicit QUBO has lower written-term order than the
$\Theta(N^3L^2)$ flow expansion. For constant $L$, both are
$\Theta(N^3)$ and there is no uniform exact ordering. The flow construction's
engineering advantage is instead implementation dependent: its balance
residuals can be evaluated and uncomputed one at a time while reusing the same
workspace, without materializing their dense pairwise expansion.
The position bound in the table gives the same arithmetic treatment to its
one-hot and counting penalties but retains explicit phases for its transition
table. It is an upper bound, not a lower bound; an adjacency or cost oracle
could change the comparison, but its memory preparation and query resources
would then also have to be counted.

On a sparse directed graph, the flow register and occurrence count track
$|E|$ and the node degrees, whereas the displayed position QUBO pays for
forbidden transitions as well as allowed ones. The single-sequence position
representation also retains route-block permutation symmetry, while the arc
variables describe an unordered route set; the capped flow registers can
introduce their own
auxiliary degeneracy. Coefficient dynamic range depends on the penalty and
cost scales and cannot be ranked from asymptotic counts alone. The engineering
contribution here is an explicit qubit--coupler--gate trade-off for the two
concrete implementations under a common logical accounting model; alternative
oracle constructions lead to different resource balances.

\iflongversion
\section{Position-Indexed Alternative}
\label{subsec:solving-vrp}

We next give a single-sequence position-indexed formulation of the same
fixed-fleet, homogeneous, uncapacitated VRP. Related QUBO
routing models
evaluated on quantum or quantum-inspired hardware include
Refs.~\cite{Harwood2021,Masuda2023,Leonidas2024QubitEfficient,Suen2022,Bialczak2026,Onah2026}.
In particular, Bia{\l}czak et al.~\cite{Bialczak2026} use
vehicle--position--customer variables for equal-demand CVRP. Under their
relation $n=mQ$, that direct register has $mQn=n^2$ problem variables, whereas
the single-sequence construction below uses one binary variable per
vertex--global-position pair. The index $t$ records a position in a cyclic
visit order, not physical time.
The Hamiltonian below does not itself encode capacities, travel times, service
times, or time windows. Such extensions remain possible, but resetting time or
load at depot separators requires boundary-aware constraints. An explicit
vehicle index makes route-local timing, capacity, and heterogeneous-fleet
constraints more direct, at the cost of the parameter-dependent $KNS$ register
described above. Unlike the two-index arc model, consecutive positions encode a
connected walk directly.

For \(K=1\), the VRP reduces to finding a minimum-cost Hamiltonian cycle, i.e.,
the TSP with a fixed starting node (the depot), which does not affect the total
cost.

Because vehicle identities do not affect the homogeneous objective, we
concatenate the $K$ route blocks into a single cyclic sequence. Depot
occurrences separate successive routes; permuting complete route blocks
represents the same unordered collection of routes.

We place the depot, denoted by node $v=1$, at position $t=1$ of the cyclic
sequence. Each subsequent position records the next visited vertex. We introduce
binary decision variables \(x_{v,t}\), where
\[
x_{v,t}=
\begin{cases}
1, & \text{if node \(v\) occupies position \(t\),}\\
0, & \text{otherwise.}
\end{cases}
\]
Nodes \(v=2,3,\ldots,N\) represent customers. The sequence contains each of the
$N-1$ customers once and contains the depot $K$ times, so its length is
\[
T = N + K - 1,
\]
and we impose a cyclic boundary condition on the position index, namely
\begin{equation}
\label{eq:time-periodic}
x_{v,T+1}\equiv x_{v,1}\qquad \text{for all } v\in\{1,\ldots,N\}.
\end{equation}
Equivalently, throughout the paper we interpret $t+1$ modulo $T$, so that the
transition term $x_{u,t}x_{v,t+1}$ for $t=T$ is understood as $x_{u,T}x_{v,1}$.

The VRP is expressed as the following quadratic optimization problem.

\noindent\textbf{Minimize:}
\begin{align}
\label{eq:vobj}
\sum_{(u,v)\in E} c_{uv}\sum_{t=1}^{T} x_{u,t}x_{v,t+1}.
\end{align}

\noindent\textbf{Subject to:}
\begin{align}
\label{eq:vcon1}
&\sum_{t=2}^{T} x_{v,t}=1,
\qquad && v=2,3,\ldots,N,\\
\label{eq:vcon2}
&\sum_{v=1}^{N} x_{v,t}=1,
\qquad && t=1,2,\ldots,T,\\
\label{eq:vcon3}
&x_{1,1}=1,\\
\label{eq:vcon4}
&\sum_{t=1}^{T} x_{1,t}=K,\\
\label{eq:vcon5}
&\sum_{(u,v)\in\overline E} x_{u,t}x_{v,t+1}=0,
\qquad && t=1,2,\ldots,T,\\
\label{eq:vcon6}
&x_{1,t}x_{1,t+1}=0,
\qquad && t=1,2,\ldots,T.
\end{align}
Constraint~\eqref{eq:vcon1} ensures that each customer is visited exactly once,
Constraint~\eqref{eq:vcon2} enforces that exactly one node occupies each
position, Constraint~\eqref{eq:vcon3} fixes the initial position at the depot,
Constraint~\eqref{eq:vcon4} enforces exactly \(K\) depot occurrences, which
serve as separators between the \(K\) cyclic route blocks, and
Constraint~\eqref{eq:vcon5} forbids transitions between distinct
vertices that are not joined by an allowed arc. Constraint~\eqref{eq:vcon6} is
the nonempty-route condition: because the $K$ depot
occurrences partition the cyclic sequence into $K$ route segments, forbidding
consecutive depot occurrences ensures that every segment contains at least one
customer. Together, these constraints exclude both forbidden transitions and
disconnected customer-only subtours.

Similarly to the TSP case, we construct the quadratic Hamiltonian
\begin{equation}
\label{eq:HVRP}
\begin{split}
H_{\mathrm{VRP}}
&= A\left(1-x_{1,1}\right)^2
+ A\sum_{v=2}^{N}\left(1-\sum_{t=2}^{T}x_{v,t}\right)^2\\
&\quad+A\sum_{t=1}^{T}\left(1-\sum_{v=1}^{N}x_{v,t}\right)^2\\
&\quad+A\left(K-\sum_{t=1}^{T}x_{1,t}\right)^2\\
&\quad+A\sum_{(u,v)\in\overline E}\sum_{t=1}^{T}
x_{u,t}x_{v,t+1}\\
&\quad+A_e\sum_{t=1}^{T}x_{1,t}x_{1,t+1}\\
&\quad+B\sum_{(u,v)\in E}c_{uv}\sum_{t=1}^{T}
x_{u,t}x_{v,t+1},
\end{split}
\end{equation}
where $A,A_e>0$ are penalty weights and $B>0$ scales the objective. The
term with coefficient $A_e$ penalizes every pair of
consecutive depot occurrences, including the cyclic boundary. Let
$P_V=\min\{A,A_e\}$ and
$c_{\max}=\max_{(u,v)\in E}c_{uv}$. If
\[
P_V \;>\; T B c_{\max},
\]
then any assignment that violates at least one of the constraints
in~\eqref{eq:vcon1}--\eqref{eq:vcon6} has energy strictly larger than every
feasible assignment. Indeed, a feasible sequence contains exactly $T$
transitions and therefore costs at most $BTc_{\max}$, whereas every violated
integer equality, forbidden transition, or consecutive-depot condition incurs
penalty at least $P_V$. A zero-penalty cyclic sequence contains each customer
once, contains the depot $K$ times with no consecutive occurrences, and uses
only allowed arcs. It consequently decomposes into exactly $K$ nonempty
depot-to-depot routes. Thus every ground-state routing assignment is feasible
and cost-optimal.

When $K=1$, the feasible VRP routes and their objective values coincide with
the fixed-depot TSP tours, but the unreduced QUBO polynomials are not identical
on infeasible bitstrings. The exact equality after forced-variable elimination
is displayed after both reduced Hamiltonians in Section~\ref{sec:savingQbits}.

The number of binary variables, and hence logical problem qubits in the direct
Ising encoding, is
\[
n_{\mathrm{prob}}^{\mathrm{pos}} = N\cdot T = N(N+K-1),
\]
since we have one variable \(x_{v,t}\) for each \(v\in\{1,\ldots,N\}\) and
\(t\in\{1,\ldots,T\}\). Since $1\le K\le N-1$ gives
$N\le T\le2N-2$, this scales uniformly as
\(n_{\mathrm{prob}}^{\mathrm{pos}}=\Theta(N^2)\).
As in the flow model, this count excludes embedding, routing, and
fault-tolerance overhead.

The transition part contains the allowed-arc objective, the forbidden-arc
penalty, and the explicit nonempty-route penalty. Let
$E_+=\{(u,v)\in E:c_{uv}>0\}$. Because
$|E_+|+|\overline E|\le N(N-1)$, the number of nonzero transition-monomial
occurrences is at most
\[
T|E_+|+T|\overline E|+T
\le T\bigl(N(N-1)+1\bigr)
=O(TN^2),
\]
with worst-case order $\Theta(TN^2)$. Thus sparsity of the allowed graph alone
does not imply $\Theta(T|E|)$ transition complexity: a smaller allowed-arc set
is offset by a larger forbidden-arc set.

The squared one-hot penalties contribute additional quadratic occurrences. Before
collecting repeated couplers, their total together with the transition terms is
\[
(N-1)\binom{T-1}{2}
+T\binom{N}{2}
+\binom{T}{2}
+T\bigl(|E_+|+|\overline E|\bigr)
+T.
\]
Consequently, the full explicit QUBO contains
$\Theta(NT^2+TN^2)$ quadratic-monomial occurrences. Since
$T=N+K-1=\Theta(N)$ for $1\le K\le N-1$, this becomes $\Theta(N^3)$.

After the standard map $x\mapsto(1-\sigma^z)/2$, a direct gate-based phase
separator can apply one two-qubit phase operation per collected nonzero
quadratic coupler. A direct controlled-phase or $ZZ$ decomposition needs no
ancilla; an AND--phase--uncompute pattern may instead use one reusable ancilla.
This direct implementation has the same $\Theta(NT^2+TN^2)$ asymptotic
logical-gate count. A sparse-graph improvement would require a different
adjacency-oracle or constraint implementation and does not follow from the QUBO
written here.

For the reusable-workspace row of
Table~\ref{tab:flow-position-tradeoff}, we instead apply the same arithmetic
construction to the squared counting penalties. For each customer-count,
position-occupancy, or depot-count constraint, a signed $O(\log N)$-qubit
accumulator forms the residual, the polynomial
in Eq.~\eqref{eq:signed-square-phase} applies its square phase directly on
that accumulator, and the residual computation is uncomputed and reused.
The $N-1$
customer-count residuals require
$O(NT\log N+N\log^2N)$ gates, the $T$ position-occupancy residuals require
$O(TN\log N+T\log^2N)$ gates, and the depot-count residual requires
$O(T\log N+\log^2N)$ gates. Because $T=\Theta(N)$, all squared counting
penalties together require $O(NT\log N)$ gates and $O(\log N)$ reusable
workspace. The transition, forbidden-arc, objective, and
consecutive-depot terms remain explicit and require $O(TN^2)$ phases in the
dense worst case. This gives the position-model bound
$O(TN^2+NT\log N)$ reported in the table. Thus compute--phase--uncompute
improves the arithmetic part of both formulations, but in the position model
the explicit transition table still determines the leading order.

\subsection{Forced-Variable Reductions}
\label{sec:savingQbits}

Following the method demonstrated in~\cite{Gurevich2025TSP}, we use tables to
show how logical problem-variable reduction is achieved in the VRP formulation.
The number of binary variables is reduced from $NT$ to $N(T-1)$ by fixing
\begin{equation}
\label{eq:trivialVRP}
x_{v,1}=
\begin{cases}
1 & v=1, \\
0 & v\neq 1,
\end{cases}
\end{equation}
since the concatenated route sequence begins at the depot. This saves exactly
$N$ variables. Because $x_{1,1}=1$ and the nonempty-route constraint
\eqref{eq:vcon6} imply that the adjacent positions cannot also contain the
depot, we may additionally fix
\begin{equation}
\label{eq:trivialVRPextra}
x_{1,2}=x_{1,T}=0,
\end{equation}
giving $N(T-1)-2$ variables when $T\ge3$. For the exceptional case $T=2$,
these symbols denote the same variable and only one additional variable is
removed. The two boundary fixings do not themselves enforce nonempty routes;
that property follows from the no-consecutive-depot penalty at every cyclic
pair of positions.

The exact enumerations below use only the first reduction in
Eq.~\eqref{eq:trivialVRP}. Both $NT$ and $N(T-1)$ remain $\Theta(N^2)$ over
the full range $1\le K\le N-1$; the saving is practical at small sizes rather
than asymptotic.

\begin{table}[!t]
\caption{Table representation of variables in $H_{\mathrm{VRP}}$ from
Eq.~\eqref{eq:HVRP}. Fixed variables are shown in red, while free
optimization variables are marked by $*$.}
\label{tab:reducedVRP}
\centering
\renewcommand{\arraystretch}{1.5}
\begin{tabular}{|c|c|c|c|c|c|}
\hline
v \(\setminus\) t&
\textbf{1} & \textbf{2} & \textbf{3} & \textbf{$\cdots$} & \textbf{T} \\
\hline
\textbf{1} &
\cellcolor{red!25}\textbf{1} &
\cellcolor{green!25}$*$ &
\cellcolor{green!25}$*$ &
\cellcolor{green!25}$\cdots$ &
\cellcolor{green!25}$*$ \\
\hline
\textbf{2} &
\cellcolor{red!25}\textbf{0} &
\cellcolor{green!25}$*$ &
\cellcolor{green!25}$*$ &
\cellcolor{green!25}$\cdots$ &
\cellcolor{green!25}$*$ \\
\hline
\textbf{3} &
\cellcolor{red!25}\textbf{0} &
\cellcolor{green!25}$*$ &
\cellcolor{green!25}$*$ &
\cellcolor{green!25}$\cdots$ &
\cellcolor{green!25}$*$ \\
\hline
\textbf{$\vdots$} &
\cellcolor{red!25}$\vdots$ &
\cellcolor{green!25}$\vdots$ &
\cellcolor{green!25}$\vdots$ &
\cellcolor{green!25}$\ddots$ &
\cellcolor{green!25}$\vdots$ \\
\hline
\textbf{N} &
\cellcolor{red!25}\textbf{0} &
\cellcolor{green!25}$*$ &
\cellcolor{green!25}$*$ &
\cellcolor{green!25}$\cdots$ &
\cellcolor{green!25}$*$ \\
\hline
\end{tabular}
\end{table}

These substitutions are logical consequences of the constraints and therefore
preserve the feasible route set and every feasible objective value. They reduce
the finite search domain without changing the asymptotic register size.

Substituting the fixed values from Eq.~\eqref{eq:trivialVRP} into Eq.~\eqref{eq:HVRP} and rearranging terms yields the reduced Hamiltonian, which reflects the reduced set of variables:
\ifdefined\IEEEBUILD
\begin{equation}
\label{eq:redHVRP}
\begin{aligned}
H_{\mathrm{VRP}}^{\text{reduced}} &=
A\Bigg[
\sum_{v=2}^{N} \left( 1 - \sum_{t=2}^{T} x_{v,t} \right)^2 \\
&\qquad
+ \sum_{t=2}^{T} \left( 1 - \sum_{v=1}^{N} x_{v,t} \right)^2 \\
&\qquad
+ \left( (K-1) - \sum_{t=2}^{T} x_{1,t} \right)^2
\Bigg] \\
&\quad + A\Bigg[
\sum_{(u,v)\in\overline E} \sum_{t=2}^{T-1} x_{u,t} x_{v,t+1} \\
&\qquad + \sum_{(1,v)\in\overline E} x_{v,2}
+ \sum_{(u,1)\in\overline E} x_{u,T}
\Bigg] \\
&\quad + A_e\Bigg(
x_{1,2}+x_{1,T}
+\sum_{t=2}^{T-1}x_{1,t}x_{1,t+1}
\Bigg) \\
&\quad + B\Bigg[
\sum_{(u,v)\in E} c_{uv} \sum_{t=2}^{T-1} x_{u,t} x_{v,t+1} \\
&\qquad + \sum_{(1,v)\in E} c_{1v} x_{v,2}
+ \sum_{(u,1)\in E} c_{u1} x_{u,T}
\Bigg].
\end{aligned}
\end{equation}
\else
\begin{equation}
\label{eq:redHVRP}
\begin{split}
H_{\mathrm{VRP}}^{\text{reduced}} &=
A\Bigg[
\sum_{v=2}^{N} \left( 1 - \sum_{t=2}^{T} x_{v,t} \right)^2
+ \sum_{t=2}^{T} \left( 1 - \sum_{v=1}^{N} x_{v,t} \right)^2 \\
&\qquad\quad
+ \left( (K-1) - \sum_{t=2}^{T} x_{1,t} \right)^2
\Bigg] \\
&\quad + A\Bigg[
\sum_{(u,v)\in\overline E} \sum_{t=2}^{T-1} x_{u,t} x_{v,t+1}
+ \sum_{(1,v)\in\overline E} x_{v,2}
+ \sum_{(u,1)\in\overline E} x_{u,T}
\Bigg] \\
&\quad + A_e\left(
x_{1,2}+x_{1,T}
+\sum_{t=2}^{T-1}x_{1,t}x_{1,t+1}
\right) \\
&\quad + B\Bigg[
\sum_{(u,v)\in E} c_{uv} \sum_{t=2}^{T-1} x_{u,t} x_{v,t+1}
+ \sum_{(1,v)\in E} c_{1v} x_{v,2}
+ \sum_{(u,1)\in E} c_{u1} x_{u,T}
\Bigg].
\end{split}
\end{equation}
\fi

For the TSP, since there are no intermediate visits at the starting point, we may further fix
\begin{equation}
\label{eq:trivialTSP}
x_{1,t}=
\begin{cases}
1 & t=1, \\
0 & t\neq 1,
\end{cases}
\end{equation}
which reduces the number of variables from $N^2$ to $(N-1)^2$. The fixed
variables are summarized in Table~\ref{tab:reducedTSP}.

\begin{table}[!t]
    \caption{Table representation of variables in $H_{\mathrm{TSP}}$ from
    Eq.~\eqref{eq:HTSP_fixed}. Fixed variables are shown in red, while free
    optimization variables are marked by $*$.}
    \label{tab:reducedTSP}
    \centering
    \renewcommand{\arraystretch}{1.5}
    \begin{tabular}{|c|c|c|c|c|c|}
        \hline
        v \(\setminus\) t &
        \textbf{1} & \textbf{2} & \textbf{3} & \textbf{$\cdots$} & \textbf{N} \\
        \hline
        \textbf{1} &
        \cellcolor{red!25}\textbf{1} &
        \cellcolor{red!25}\textbf{0} &
        \cellcolor{red!25}\textbf{0} &
        \cellcolor{red!25}$\cdots$ &
        \cellcolor{red!25}\textbf{0} \\
        \hline
        \textbf{2} &
        \cellcolor{red!25}\textbf{0} &
        \cellcolor{green!25}$*$ &
        \cellcolor{green!25}$*$ &
        \cellcolor{green!25}$\cdots$ &
        \cellcolor{green!25}$*$ \\
        \hline
        \textbf{3} &
        \cellcolor{red!25}\textbf{0} &
        \cellcolor{green!25}$*$ &
        \cellcolor{green!25}$*$ &
        \cellcolor{green!25}$\cdots$ &
        \cellcolor{green!25}$*$ \\
        \hline
        \textbf{$\vdots$} &
        \cellcolor{red!25}$\vdots$ &
        \cellcolor{green!25}$\vdots$ &
        \cellcolor{green!25}$\vdots$ &
        \cellcolor{green!25}$\ddots$ &
        \cellcolor{green!25}$\vdots$ \\
        \hline
        \textbf{N} &
        \cellcolor{red!25}\textbf{0} &
        \cellcolor{green!25}$*$ &
        \cellcolor{green!25}$*$ &
        \cellcolor{green!25}$\cdots$ &
        \cellcolor{green!25}$*$ \\
        \hline
    \end{tabular}
\end{table}

Substituting Eq.~\eqref{eq:trivialTSP} into
Eq.~\eqref{eq:HTSP_fixed} yields the reduced TSP Hamiltonian:
\ifdefined\IEEEBUILD
\begin{equation}
\label{eq:redHTSP}
\begin{aligned}
H_{\mathrm{TSP}}^{\text{reduced}} &=
A\Bigg[
\sum_{v=2}^{N} \left( 1 - \sum_{t=2}^{N} x_{v,t} \right)^2 \\
&\qquad + \sum_{t=2}^{N} \left( 1 - \sum_{v=2}^{N} x_{v,t} \right)^2
\Bigg] \\
&\quad + A\Bigg[
\sum_{(u,v)\in\overline E} \sum_{t=2}^{N-1} x_{u,t} x_{v,t+1} \\
&\qquad + \sum_{(1,v)\in\overline E} x_{v,2}
+ \sum_{(u,1)\in\overline E} x_{u,N}
\Bigg] \\
&\quad + B\Bigg[
\sum_{(u,v)\in E} c_{uv} \sum_{t=2}^{N-1} x_{u,t} x_{v,t+1} \\
&\qquad + \sum_{(1,v)\in E} c_{1v} x_{v,2}
+ \sum_{(u,1)\in E} c_{u1} x_{u,N}
\Bigg].
\end{aligned}
\end{equation}
\else
\begin{equation}
\label{eq:redHTSP}
\begin{split}
H_{\mathrm{TSP}}^{\text{reduced}} &=
A\Bigg[
\sum_{v=2}^{N} \left( 1 - \sum_{t=2}^{N} x_{v,t} \right)^2
+ \sum_{t=2}^{N} \left( 1 - \sum_{v=2}^{N} x_{v,t} \right)^2
\Bigg] \\
&\quad + A\Bigg[
\sum_{(u,v)\in\overline E} \sum_{t=2}^{N-1} x_{u,t} x_{v,t+1}
+ \sum_{(1,v)\in\overline E} x_{v,2}
+ \sum_{(u,1)\in\overline E} x_{u,N}
\Bigg] \\
&\quad + B\Bigg[
\sum_{(u,v)\in E} c_{uv} \sum_{t=2}^{N-1} x_{u,t} x_{v,t+1}
+ \sum_{(1,v)\in E} c_{1v} x_{v,2}
+ \sum_{(u,1)\in E} c_{u1} x_{u,N}
\Bigg].
\end{split}
\end{equation}
\fi

For $K=1$, eliminating the forced variables
$x_{1,t}=0$ for $t=2,\ldots,N$ and $x_{v,1}=0$ for $v\ge2$ gives
\[
\left.
H_{\mathrm{VRP}}^{\mathrm{reduced}}
\right|_{\substack{
K=1\\
x_{1,t}=0,\;t=2,\ldots,N
}}
=H_{\mathrm{TSP}}^{\mathrm{reduced}}.
\]

These substitutions do not change the $\Theta(N^2)$ asymptotic register size,
but they matter at the small sizes accessible to exact enumeration. For the
four-node TSP checked below, the reduction from $16$ to $9$ variables decreases
the enumeration domain from $2^{16}$ to $2^9$ assignments.

\fi %

\section{Exact Verification and Hardware Characterization on Small Instances}
\label{sec:experiments}

We use four complementary checks. Exact basis-state enumeration validates the
stated finite QUBOs; a deterministic local circuit audit checks the
compute--phase--uncompute logic; a systematic classical benchmark measures the
connectivity gap; and a shallow hardware experiment characterizes termwise
Ising circuits for the four-node counterexample. These checks answer different
questions and do not provide evidence of quantum advantage. The enumeration,
circuit audit, and connectivity benchmark are classical checks;
Sec.~\ref{subsec:iqm-hardware} is the only QPU experiment, and its circuits use
the termwise Ising realization rather than the product-free arithmetic oracle.

\subsection{Exhaustive classical check of the flow correction}

Consider the directed graph on four vertices with depot $1$, one vehicle, and
\[
E=\{(1,2),(2,1),(3,4),(4,3),(2,3),(4,1)\}.
\]
Give the first four listed arcs cost $1$ and the last two cost $10$.  The local
degree equations have exactly two solutions.  The cheaper one,
\[
\{(1,2),(2,1),(3,4),(4,3)\},
\]
has cost $4$ but contains the customer-only cycle $3\to4\to3$.  The other is
the connected route
\[
1\to2\to3\to4\to1
\]
of cost $22$.

For the corrected model, $U=N-K=3$, $L=2$, and the register contains six
routing bits and twelve flow bits.  We set $B=1$ and
$A=A_f=A_c=41$, one unit above the sufficient objective bound
$(N-1+K)c_{\max}=40$.  Exact evaluation of all $2^{18}=262144$ assignments
gives minimum energy $22$ and one minimizing bitstring. This enumerates the
diagonal energy on every 18-bit assignment; it neither simulates a quantum
circuit nor uses hardware. The minimizer selects the connected route and
carries flows $3,2,1,0$ on arcs
$(1,2),(2,3),(3,4),(4,1)$, respectively. Thus the finite check directly
distinguishes the disconnected degree-only ground state from the connected
flow-augmented ground state.

\iflongversion
To test the independent position-model reductions, we also enumerate three
reduced registers on two separate complete undirected instances. The
four-node, one-vehicle instance is encoded both as a fixed-depot TSP and as a
position-indexed VRP; the third register encodes a three-node, two-vehicle VRP.
These complete instances are not the sparse flow-correction instance above, so
their minimum energies should not be compared with the value $22$. The
four-node encodings use $A=41$ and $A=A_e=41$, respectively, and the
three-node encoding uses $A=A_e=37$. The no-consecutive-depot term is included
in every position-indexed VRP energy. The results are summarized in
Table~\ref{tab:exact-enumeration}.
\begin{table*}[t]
\centering
\scriptsize
\setlength{\tabcolsep}{3pt}
\caption{Exact computational-basis enumeration. The first row is the sparse
flow diagnostic. The remaining three rows concern two separate complete
undirected instances; the two $N=4$, $K=1$ rows are alternative encodings of
the same routing instance. Degeneracy counts distinct binary assignments in
the indicated encoding.}
\label{tab:exact-enumeration}
\begin{tabular}{lcccccc}
\toprule
Instance & \(N\) & \(K\) & \(n\) & Penalty weights & Minimum & Degeneracy\\
\midrule
Sparse flow diagnostic & 4 & 1 & 18 & \(A=A_f=A_c=41\) & \(22\) & 1\\
Complete fixed-depot TSP & 4 & 1 & 9  & \(A=41\) & \(16\) & 2\\
Complete position VRP & 4 & 1 & 12 & \(A=A_e=41\) & \(16\) & 2\\
Complete position VRP & 3 & 2 & 9  & \(A=A_e=37\) & \(10\) & 2\\
\bottomrule
\end{tabular}
\end{table*}
\fi

\iflongversion
\subsection{Four-node single-vehicle instance}

Consider the complete undirected graph on four vertices, represented by both
orientations of every edge, with
\[
\begin{gathered}
c_{12}=1,\quad c_{13}=1,\quad c_{14}=10,\\
c_{23}=3,\quad c_{24}=9,\quad c_{34}=5.
\end{gathered}
\]
Up to reversal, its three Hamiltonian cycles have costs \(16\), \(19\), and
\(23\). The unique minimum-cost undirected cycle therefore has the two
orientations
\[
1\to2\to4\to3\to1,
\qquad
1\to3\to4\to2\to1,
\]
each of cost \(16\).

After the TSP depot variables are eliminated, the reduced register contains
\((N-1)^2=9\) bits. Exact evaluation of all \(2^9=512\) assignments gives
minimum energy \(16\) and precisely the two bitstrings encoding the displayed
orientations. For the position-indexed VRP with \(K=1\), the first-column
reduction leaves \(N(T-1)=12\) bits. Exact evaluation of all \(2^{12}=4096\)
assignments again gives minimum energy \(16\) with two minimizers. These two
models agree on their feasible routes and objective values, although their
energies on infeasible bitstrings need not agree before the additional depot
variables are eliminated.

\subsection{Three-node two-vehicle instance}

Now take the complete undirected graph on vertices \(\{1,2,3\}\), with depot
\(1\), \(K=2\), and
\[
c_{12}=1,\qquad c_{13}=4,\qquad c_{23}=9.
\]
Because both routes must be nonempty and there are only two customers, every
feasible solution assigns one customer to each route. Hence the route
collection is necessarily
\[
\{\,1\to2\to1,\;1\to3\to1\,\}
\]
and has total cost
\[
2c_{12}+2c_{13}=10.
\]
The first-column reduction leaves \(N(T-1)=9\) bits. Evaluation of all
\(2^9=512\) assignments gives minimum energy \(10\) and two minimizing
bitstrings, corresponding to the two possible orders of the route blocks in
the depot-delimited cyclic sequence:
\[
1\to2\to1\to3\to1,
\qquad
1\to3\to1\to2\to1.
\]
For indistinguishable vehicles these bitstrings encode the same unordered
route collection. The twofold degeneracy belongs to the ordering of the route
blocks, not to two different routing costs.
\fi

\subsection{Functional audit of compute--phase--uncompute}

The reproducibility package contains a deterministic classical audit of the
reversible flow-penalty schedule on the same 18-variable sparse instance. The
archived reference uses 23 logical wires: 18 persistent problem wires and a
signed five-qubit residual accumulator, together with ancilla-free
high-control modular increments. An optimized ripple-carry realization may
instead require $O(q)$ reusable carry wires, already included in
Eq.~\eqref{eq:flow-peak-qubits}. Direct evaluation agrees with the QUBO energy
on all $2^{18}=262144$ basis assignments. The product-free square phase was
checked on all 32 signed accumulator words. Phase accumulation and cleanup
were also checked on 550 distinct inputs formed from the all-zero and all-one
strings, every singleton and complement, every minimum-energy assignment, and
512 pseudorandom strings generated with seed 20260723. The problem bits were
unchanged and the residual register returned to zero. The phase produced by
the product-free residual circuit was compared with that of the fully expanded
termwise Hamiltonian on the same input. At the fixed diagnostic angle
$\gamma=0.017$, the maximum phase difference was
$1.62\times10^{-13}$ radians. This angle is used only for the classical
consistency check; it is unrelated to the optimized hardware angles
$\gamma_\star$ below and has no variational significance. The reported
difference is a numerical phase-consistency residual, not a hardware error
rate.

This is a functional local audit, not a QPU experiment or an empirical proof of
the asymptotic gate bound. In particular, the archived unoptimized generic
decomposition uses high-control modular increments and is intentionally not
the optimized ripple-carry implementation assumed in
Eq.~\eqref{eq:flow-gate-bound}; its large decomposition therefore should not be
read as a practical NISQ circuit estimate. The audited inputs, outputs, and
verification script are included in the reproducibility
package~\cite{VRPRepro2026}.

\subsection{Systematic exact connectivity benchmark}

The preceding bitstring enumeration verifies the corrected Hamiltonian on one
sparse instance. We separately measure how often the connectivity omission
changes the exact combinatorial optimum on complete graphs. For each cost
matrix and fleet size, let $z_{\mathrm{conn}}$ be the minimum cost of $K$
nonempty depot-to-depot routes and let $z_{\mathrm{disc}}$ be the minimum cost
among degree-feasible selections containing at least one customer-only cycle,
with $z_{\mathrm{disc}}=+\infty$ when no such selection exists.
Thus
\[
z_{\mathrm{degree}}=\min\{z_{\mathrm{conn}},z_{\mathrm{disc}}\},
\]
and $z_{\mathrm{disc}}<z_{\mathrm{conn}}$ means that every degree-only ground
state is disconnected. Equality means that connected and disconnected ground
states coexist.

We computed these quantities exactly by Held--Karp subset dynamic programming
followed by exact set-partition dynamic programming. The benchmark contains
$100$ deterministic instances for every $N=5,\ldots,10$ in each of two
families: symmetric metric costs obtained by rounding Euclidean distances
upward between distinct integer points, and asymmetric integer arc costs in
$\{1,\ldots,1000\}$. Instances were generated by a fixed SplitMix64 stream
with base seed $2026072000$. All values of $K=1,\ldots,N-1$ were evaluated, for a
total of $1{,}200$ cost matrices and $7{,}800$ matrix--fleet-size
evaluations. Table~\ref{tab:systematic-connectivity}
summarizes the common small-fleet cases for which $N-K\ge3$. The relative
underestimate is
$(z_{\mathrm{conn}}-z_{\mathrm{degree}})/z_{\mathrm{conn}}$ and is averaged
over all instances, including zero gaps.
The exact-enumeration scripts, benchmark data, and independent verification
records are archived in the accompanying reproducibility
package~\cite{VRPRepro2026}.

\begin{table}[!t]
\centering
\caption{Exact benchmark of the degree-only connectivity gap. ``Strict'' is
the percentage with $z_{\mathrm{disc}}<z_{\mathrm{conn}}$.}
\label{tab:systematic-connectivity}
\small
\setlength{\tabcolsep}{3.5pt}
\begin{tabular}{@{}lrrrr@{}}
\toprule
Cost family & $K$ & Cases & Strict (\%) & Mean gap (\%)\\
\midrule
Euclidean metric & 1 & 600 & 98.5 & 20.32\\
                 & 2 & 600 & 97.5 & 14.70\\
                 & 3 & 500 & 99.0 & 10.72\\
Asymmetric integer & 1 & 600 & 64.8 & 6.42\\
                   & 2 & 600 & 31.0 & 1.75\\
                   & 3 & 500 & 21.4 & 0.86\\
\bottomrule
\end{tabular}
\end{table}

There is also a sharp structural boundary. A disconnected degree-feasible
selection needs at least one customer on each of the $K$ depot routes and at
least two customers on a customer-only cycle. Hence it is impossible when
$N-K\le2$. It becomes possible already at $N-K=3$: among the $600$ metric
instances on this boundary, $99.83\%$ had a strictly disconnected optimum and
$0.17\%$ had a connected--disconnected tie; the corresponding asymmetric
figures were $5.50\%$ and $0.17\%$. These exact frequencies describe only the
stated deterministic samples. They are not a probability theorem, a
finite-penalty spectral test, a hardware experiment, or evidence of quantum
advantage.

\subsection{Termwise hardware characterization on the four-node diagnostic instance}
\label{subsec:iqm-hardware}

We use the hardware run as a diagnostic, not as a solver benchmark. On the same
sparse $N=4$, $K=1$ instance, we ask whether a shallow circuit follows the
erroneous low-energy preference of the degree-only model and how much the
provider-compiled footprint changes after adding the flow correction. Here
\emph{termwise} means that every collected one- and two-qubit Ising term is
implemented by its own phase operation. This is distinct from the structured
reversible-arithmetic oracle analyzed in Eq.~\eqref{eq:flow-gate-bound}.
The experiment has three stages: classical exhaustive preflight, noiseless
$p=1$ statevector parameter selection frozen before execution, and IQM
sampling. Only the third stage is hardware data.

The degree-only Hamiltonian~\eqref{eq:Azad_HVRP} acts on the six routing bits.
The full flow-augmented register contains those six bits and twelve flow bits.
For device execution we imposed the valid strengthening
\[
y_{21,0}=y_{21,1}=y_{41,0}=y_{41,1}=0,
\]
which removes the two-bit flow registers on arcs entering the depot. Every
feasible route has a canonical remaining-load certificate with zero flow on
its final arc into the depot, so this restriction preserves the routing
optimum while reducing the displayed 18-qubit register to 14 qubits. The
logical order is
\[
(q_0,\ldots,q_5)
=(x_{12},x_{21},x_{23},x_{34},x_{41},x_{43})
\]
for both Hamiltonians, followed in the reduced flow-augmented circuit by
\[
\begin{aligned}
(q_6,\ldots,q_{13})={}&
(y_{12,0},y_{12,1},y_{23,0},y_{23,1},\\
&y_{34,0},y_{34,1},y_{43,0},y_{43,1}).
\end{aligned}
\]
Thus the first six logical wires are route-choice qubits: setting \(q_j=1\)
selects the directed arc named by the corresponding \(x\)-variable. The
remaining eight wires form four two-bit flow words. They do not select
additional vehicle moves; instead, they encode the auxiliary flow carried by
the indicated selected arc and thereby certify depot connectivity.
Because $U=3$, each retained word is
$f_{ij}=y_{ij,0}+2y_{ij,1}$. In the archived count strings, $q_0$ is the
most-significant measured bit.

Before hardware execution, we used the same unscaled coefficients as in the exact test:
$B=1$ and $A=A_f=A_c=41$. Exhaustive evaluation of all
$2^{14}=16384$ reduced assignments verifies exact QUBO--Ising energy agreement
and retains the unique connected minimum of energy $22$; the six-qubit
degree-only Hamiltonian instead has the disconnected minimum of energy $4$.
No Hamiltonian rescaling was applied.

After substituting $b=(1-Z)/2$, write
\[
H_{\mathrm I}=C+\sum_qh_qZ_q+\sum_{q<r}J_{qr}Z_qZ_r.
\]
The depth-one QAOA state on $n\in\{6,14\}$ logical qubits is
\[
\lvert\psi(\gamma,\beta)\rangle
=
\exp\!\left(-i\beta\sum_{q=0}^{n-1}X_q\right)
\exp\!\left(-i\gamma H_{\mathrm I}\right)
\lvert+\rangle^{\otimes n}.
\]
Here $\gamma$ is the cost-phase angle and $\beta$ is the mixer angle. With
$R_Z(\theta)=\exp(-i\theta Z/2)$, a term $h_qZ_q$ is implemented by
$R_Z(2\gamma h_q)$ and a term $J_{qr}Z_qZ_r$ by
CNOT--$R_Z(2\gamma J_{qr})$--CNOT. The mixer uses $R_X(2\beta)$ on every
logical qubit. The Ising constant $C$ is omitted only because it contributes a
global phase. Thus the submitted logical circuits implement the Hamiltonians
defined above, rather than surrogate objective functions.

We numerically minimized the exact noiseless $p=1$ statevector expectation of
the unscaled logical Hamiltonian over
$\gamma\in[0,2\pi)$ and $\beta\in[0,\pi)$. For the degree-only case, a
$9\times9$ grid and four Nelder--Mead starts were cross-checked by seeded
differential evolution (seed 20260723); for the reduced flow-augmented case,
2,000 uniform points (seed 1) were followed by Nelder--Mead. The selected
pairs are
\[
(\gamma_\star,\beta_\star)
=(4.51968698260255,\;2.3549366790662956)
\]
for the degree-only Hamiltonian and
\[
(\gamma_\star,\beta_\star)
=(4.599315084785573,\;2.744715760897889)
\]
for the reduced flow-augmented Hamiltonian. These are numerical candidates, not
certified global optima, and the selected pairs were frozen before the QPU
outcomes were observed. We also evaluated
\[
(\gamma_\star\mathbin{\pm}\pi/10,\beta_\star),
\qquad
(\gamma_\star,\beta_\star\mathbin{\pm}\pi/20).
\]
Each shift changes one coordinate by 5\% of its search-domain period while
holding the other coordinate fixed.

The 16 circuits were submitted through Amazon Braket as one ProgramSet, that
is, one batched provider task, to IQM Emerald in region
\texttt{eu-north-1} on 23 July 2026 using Braket SDK 1.124.1. For each
Hamiltonian, the batch contained a measured uniform $p=0$ reference, the
selected $p=1$ circuit, four $p=1$ sensitivity circuits, and prepare-and-measure
$\lvert0\rangle^{\otimes n}$ and $\lvert1\rangle^{\otimes n}$ controls. Each
circuit received $2{,}000$ shots, for $32{,}000$ raw shots in total. The
controls diagnose gross preparation/readout behavior but are not calibrations
for the QAOA histograms. No readout-error mitigation or postselection was
applied. The archived provider snapshot reports an update timestamp of
23 July 2026 at 05:28 UTC; its calibration values are treated as recorded
metadata and were not independently verified. We supplied no manual physical
mapping or provider-specific compiler settings. Hence the physical-qubit
subsets and routing locations shown in Fig.~\ref{fig:iqm-topology} were chosen
automatically by the provider compiler, not selected by us. The returned
metadata does not expose the proprietary compiler version or its
optimization settings. The archive also contains 100-shot-per-circuit
technical pilots on Rigetti Cepheus-1-108Q and IQM Emerald. They are not used
in the statistics reported here.

Provider compilation assigned the measured $p=0$ reference and the two
controls to physical mappings different from the five $p=1$ circuits. For each
Hamiltonian, however, the selected circuit and all four sensitivity circuits
share one mapping. Consequently, the $p=0$--$p=1$ contrasts below are
descriptive comparisons of the returned circuits, whereas comparisons within
the $p=1$ sensitivity family are mapping matched. The multinomial-bootstrap
intervals quantify finite-shot uncertainty conditional on the observed
histograms; they do not include mapping, calibration, temporal-drift, backend,
or parameter-selection uncertainty. Every measured bit string was evaluated
with the corresponding unscaled binary Hamiltonian. The two mean-contrast
intervals below use $100{,}000$ independent multinomial resamples; the
individual intervals in Fig.~\ref{fig:iqm-hardware} use $20{,}000$.

Before provider compilation, the selected six-qubit circuit contains six $H$,
six $R_X$, ten $R_Z$, and eight CNOT gates; the selected 14-qubit circuit
contains 14 $H$, 14 $R_X$, 46 $R_Z$, and 64 CNOT gates. The returned native
six-logical-qubit circuit touches six physical qubits and uses 16 PRX and eight
CZ gates, with native dependency depth 12 and CZ-only depth six. The
14-logical-qubit circuit touches 18 physical qubits and uses 201 PRX and 124 CZ
gates, with dependency depth 127 and CZ-only depth 66. Here PRX is IQM's native
phased-$X$ single-qubit rotation, CZ is controlled-$Z$, and the reported ASAP
(as-soon-as-possible) depth is a circuit dependency depth rather than a pulse
duration. Four physical qubits touched by the 14-qubit compilation do not
appear in the returned final logical-to-readout mapping; they are compiler
routing locations, not additional logical variables or algorithmic ancillas.
Every returned CZ lies on the saved 54-qubit, 85-coupler device graph.
Figure~\ref{fig:iqm-topology} shows the returned placement. The raw counts,
author-generated logical circuits, derived compilation metrics, topology
figures, and offline analyses are archived in the reproducibility
package~\cite{VRPRepro2026}. For the provider-returned compiled circuits and
device snapshots, the public package contains hash commitments rather than
the files themselves, whose redistribution rights have not been confirmed.

For direct reading of the figure, the left panel uses
\[
(\mathrm{L}0,\ldots,\mathrm{L}5)
=(x_{12},x_{21},x_{23},x_{34},x_{41},x_{43}).
\]
The right panel retains these six route-choice labels and adds
\[
\begin{aligned}
(\mathrm{L}6,\mathrm{L}7)&=(y_{12,0},y_{12,1}),&
(\mathrm{L}8,\mathrm{L}9)&=(y_{23,0},y_{23,1}),\\
(\mathrm{L}10,\mathrm{L}11)&=(y_{34,0},y_{34,1}),&
(\mathrm{L}12,\mathrm{L}13)&=(y_{43,0},y_{43,1}).
\end{aligned}
\]
Within each pair, the first and second qubits have weights one and two,
respectively. Thus every blue label identifies a precise VRP variable, whereas
a yellow \(\mathrm{R}\) label has no VRP-variable role and marks only a
physical location used by compiler routing.

\begin{figure*}[!t]
\centering
\includegraphics[width=0.82\textwidth,alt={Two schematic IQM Emerald connectivity maps show the compiler placements of the degree-only and reduced flow circuits. Logical nodes are labeled L, routing-only nodes R, and highlighted edges are used couplers. The circuits touch six and eighteen physical qubits and contain eight and 124 CZ gates, respectively.}]{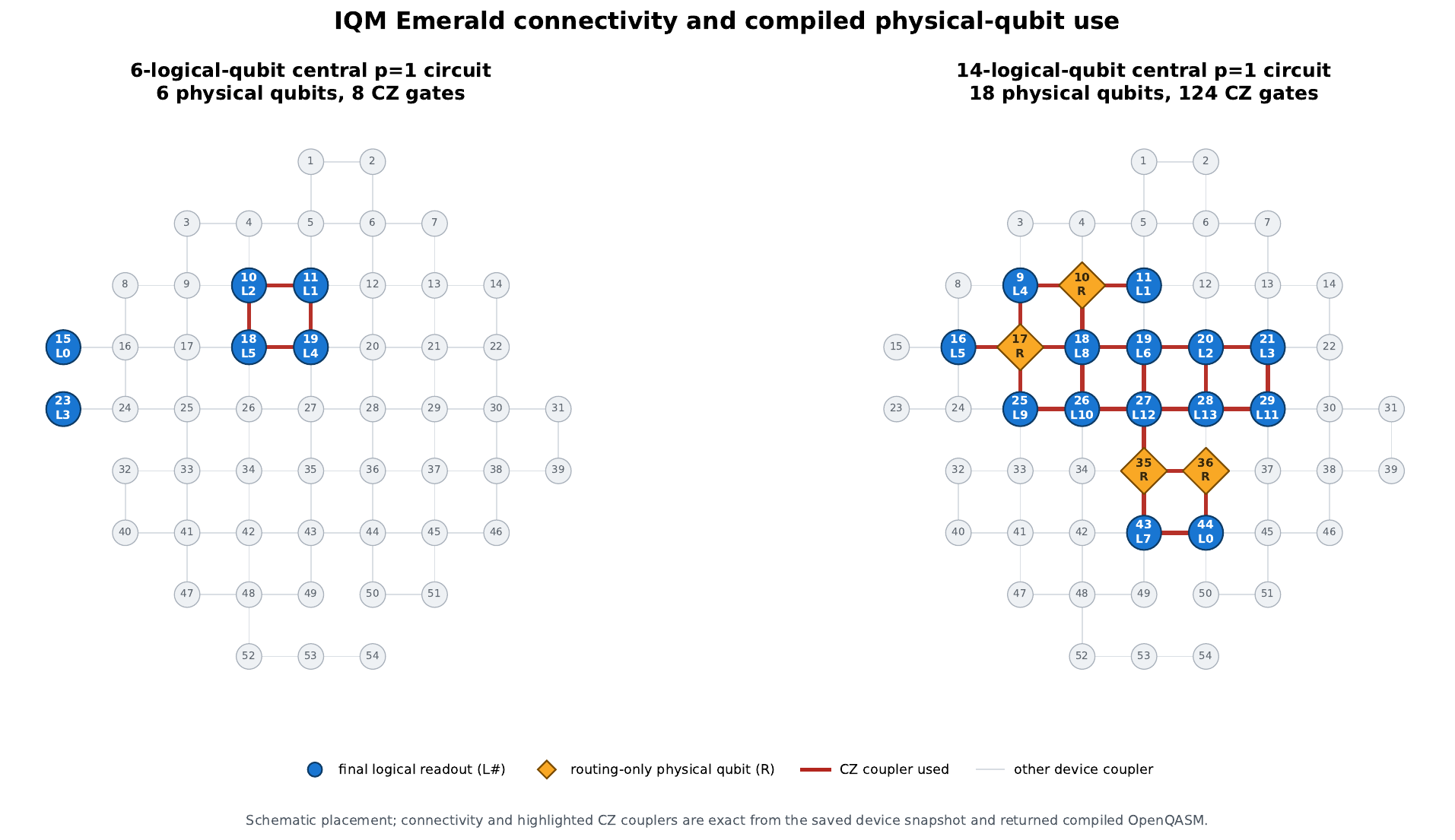}
\caption{Provider-compiled placement of the selected $p=1$ termwise-Ising
circuits on the saved IQM Emerald topology. The first line in each node is the
physical-qubit index. On a blue node, $\mathrm{L}j$ denotes logical wire $q_j$
under the variable dictionary stated immediately before the figure; the
prefix $\mathrm{L}$ means ``logical'' and is unrelated to the flow-word length
$L$. On a yellow node, $\mathrm{R}$ denotes a physical
compiler-routing location absent from the final readout mapping, not a route
variable. Each red edge is a distinct physical coupler used at least once; the
number in the panel title counts total CZ-gate occurrences. The word
``central'' inside the archived graphic refers to the selected angle pair.
The placement is the returned final mapping, not a time-resolved trace of
logical-wire motion. Node placement on the page is schematic; the saved
connectivity and highlighted couplers are exact.}
\label{fig:iqm-topology}
\end{figure*}

We classify each raw bit string before averaging its energy. In the degree-only
case, the mutually exclusive classes are: violation of a local degree/depot
constraint; satisfaction of those constraints but disconnection from the
depot; and a connected route. In the reduced flow-augmented case, the connected
class is further split according to whether the retained flow bits form a valid
balance-and-support certificate. Thus ``fully feasible'' always means a
connected VRP route; for the flow-augmented Hamiltonian it additionally
requires a valid retained-flow certificate.

Two baselines must be kept distinct. The symbol
$\bar E_{p=0}^{\mathrm{hw}}$ denotes the mean of the measured $p=0$ histogram,
whereas
\[
E_{\mathrm{unif}}=2^{-n}\operatorname{Tr}(H_{\mathrm I})
\]
is the exact uniform-distribution expectation used only to normalize
Fig.~\ref{fig:iqm-hardware}. Because the two Hamiltonians have different
energy landscapes, their raw mean energies are not compared with each other;
all contrasts are within one Hamiltonian.

For the degree-only Hamiltonian, the sampled means were
\[
\bar E_{p=0}^{\mathrm{hw}}=175.5600,
\qquad
\bar E_{p=1}^{\mathrm{hw}}=24.3480.
\]
Their descriptive contrast is $151.2120$, with conditional 95\%
multinomial-bootstrap interval $[147.6040,154.8260]$. Of the $2{,}000$
selected $p=1$ shots, $1{,}561$ ($78.05\%$) were the degree-feasible but
disconnected ground state, whereas only one shot ($0.05\%$) encoded a
connected route. The exact ideal, uncompiled logical statevector assigns
$91.4497\%$ probability to that disconnected ground state. The selected
circuit's low sampled mean is therefore dominated by the invalid degree-only
ground state.

For the reduced 14-qubit flow-augmented Hamiltonian, the sampled means were
\[
\bar E_{p=0}^{\mathrm{hw}}=706.3395,
\qquad
\bar E_{p=1}^{\mathrm{hw}}=507.7915,
\]
giving a descriptive contrast of $198.5480$ with conditional 95\% interval
$[176.8765,220.4025]$. The selected circuit has lower sampled mean energy than
each of its four mapping-matched sensitivity circuits, although this finite
sweep is not an optimality certificate. No fully model-feasible assignment was
observed in any of the six displayed flow-augmented circuits
($0$ of $12{,}000$ shots). In two selected-circuit shots, the six routing bits
selected the connected route, but the auxiliary flow bits were inconsistent. At the
selected parameters, the exact ideal, uncompiled logical statevector
ground-state probability is $4.32\times10^{-6}$, so an ideal
$2{,}000$-shot sample would have only about a $0.86\%$ chance of containing
that ground state.

\begin{figure*}[!t]
\centering
\begin{minipage}[t]{0.38\textwidth}
\centering
\includegraphics[width=\linewidth,alt={Panel a: normalized energy improvement for the six-qubit degree-only circuit, comparing the uniform reference, selected angles, and four angle perturbations. The selected mean is 0.882; horizontal whiskers show conditional 95 percent intervals.}]{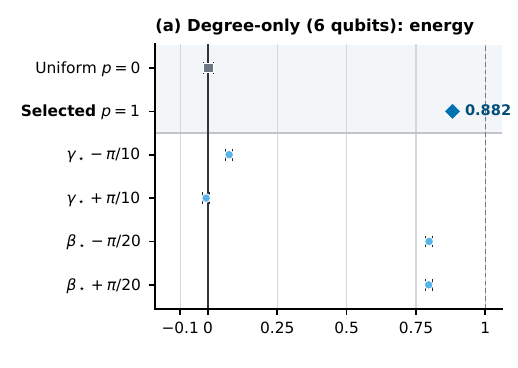}
\par\smallskip
\includegraphics[width=\linewidth,alt={Panel c: normalized energy improvement for the fourteen-qubit reduced flow circuit with the same six types of angle settings. The selected mean is 0.293; horizontal whiskers show conditional 95 percent intervals.}]{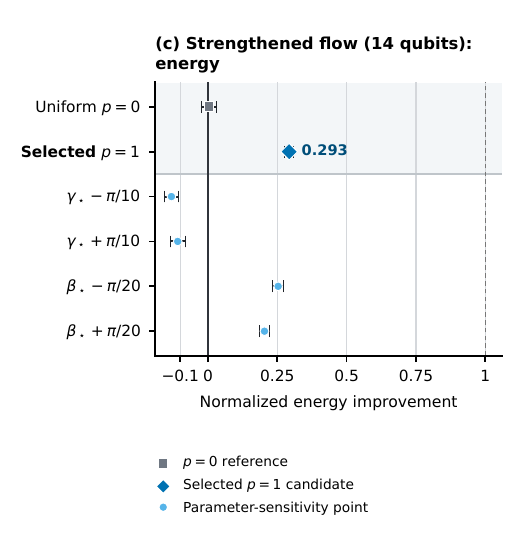}
\end{minipage}
\hspace{0.04\textwidth}
\begin{minipage}[t]{0.38\textwidth}
\centering
\includegraphics[width=\linewidth,alt={Panel b: six stacked bars partition two thousand raw shots each into infeasible, disconnected degree-feasible, and fully feasible outcomes. At the selected angles, 78.05 percent are the invalid disconnected ground state. Rows match panel a.}]{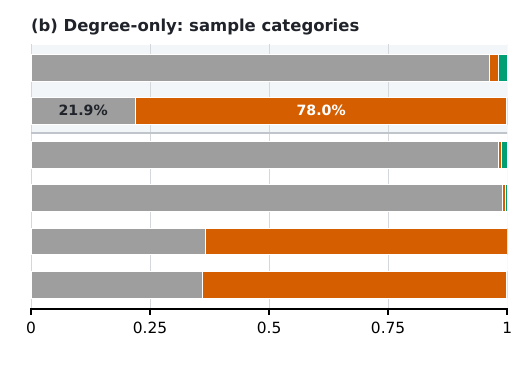}
\par\smallskip
\includegraphics[width=\linewidth,alt={Panel d: six stacked bars classify two thousand raw shots each for the reduced flow circuit, separating connected routes with invalid flow from fully feasible outcomes. There are no fully feasible outcomes in twelve thousand shots. Rows match panel c.}]{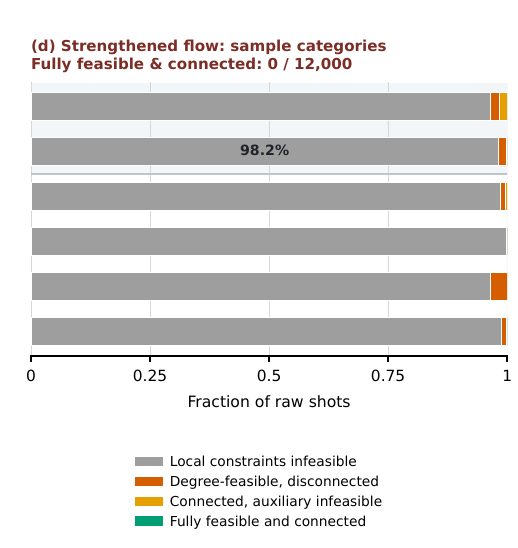}
\end{minipage}
\caption{IQM Emerald measurements for the $N=4$, $K=1$ diagnostic instance.
Panel (a) shows normalized mean energy and panel (b) the outcome fractions for
the degree-only Hamiltonian; panels (c) and (d) show the corresponding
quantities for the reduced flow-augmented Hamiltonian. Rows in each bar panel
correspond to the parameter labels in the adjacent point panel, and every bar
partitions $2{,}000$ raw shots among local-constraint violations,
degree-feasible disconnected assignments, connected assignments with invalid
retained flow (flow case only), and fully feasible assignments. The point-panel
horizontal axis is
$I=(E_{\mathrm{unif}}-\bar E)/(E_{\mathrm{unif}}-E_{\min})$: the solid line
$I=0$ is the exact uniform expectation, the dashed line $I=1$ is the exact
minimum, and $I<0$ is worse than the uniform expectation. For the degree-only
Hamiltonian, $E_{\min}$ belongs to the invalid disconnected state. Pale row
shading only groups the $p=0$ and selected rows and has no statistical meaning.
The internal title ``strengthened flow'' denotes the reduced flow-augmented
Hamiltonian described in the text.
Here ``selected'' denotes $(\gamma_\star,\beta_\star)$, with $\gamma$ the
cost-phase angle and $\beta$ the mixer angle. The labels
$\gamma_\star\mathbin{\pm}\pi/10$ and
$\beta_\star\mathbin{\pm}\pi/20$ vary only the named angle. The two controls
per Hamiltonian are not plotted. Points are measured means, horizontal
whiskers are conditional 95\% multinomial-bootstrap intervals from $20{,}000$
resamples, and every bar is a fraction of $2{,}000$ raw shots. No mitigation
or postselection was applied.}
\label{fig:iqm-hardware}
\end{figure*}

The QPU circuits implement the exact six-qubit degree-only Hamiltonian and the
explicitly reduced 14-qubit restriction of the displayed flow-augmented
Hamiltonian through a termwise Ising cost layer. They are not the
reversible-arithmetic compute--phase--uncompute circuits analyzed in the
resource section. The run therefore characterizes small-instance Hamiltonian
behavior and native compilation, but it neither tests the asymptotic gate and
workspace bounds nor demonstrates successful hardware solution of the
flow-augmented VRP instance.

\section{Conclusion and Outlook}
\label{sec:conclusion}

Local in- and out-degree penalties alone do not make an arc-based VRP
Hamiltonian connectivity correct: their zero-penalty assignments may contain
cycles disconnected from the depot. The source paper's own reported matrix
contains such a cheaper disconnected degree-feasible subgraph, and the exact benchmark in
Section~\ref{sec:experiments} shows that this is not confined to a constructed
example.

This paper contributes an explicit polynomial-size QUBO repair of the audited
arc model together with a complete ground-state guarantee. It carries the
tight fixed-fleet range $0,\ldots,N-K$
into a capped per-arc flow register, couples every flow bit to its selected
arc, proves feasibility equivalence and ground-state correctness under explicit
penalty assumptions, and separates logical problem qubits, written quadratic
interactions, and reversible-oracle resources. The displayed construction
uses
\[
|E|\bigl(1+\lceil\log_2(N-K+1)\rceil\bigr)
\]
logical problem qubits. On a complete graph its written quadratic-occurrence
count is $\Theta(N^3L^2)$, whereas workspace-reusing reversible evaluation
gives the logical upper bound $O(N^2\log N+N\log^2N)$ with $O(\log N)$
reusable workspace, without a separate product register or a materialized
dense expansion.

The flow and position encodings occupy different points of the engineering
trade-off. On a complete loopless graph, the displayed single-sequence
position model uses $NT$ problem qubits, no more than the displayed flow
register, with equality only at $K=N-1$ before forced-variable elimination.
This ordering is not generic to all position formulations: a vehicle-indexed
register with $S$ local positions per vehicle uses $KNS$ variables before
reductions and may lie on either side of the flow count. The single-sequence
model's written-occurrence count is $\Theta(N^3)$, compared with
$\Theta(N^3L^2)$ for flow; this is a lower asymptotic order when $L$ grows,
while for constant $L$ both are $\Theta(N^3)$ and the exact ordering is
parameter dependent. Under the implementations analyzed here, the position
transition table contributes $O(TN^2)$ termwise logical phases, while its
one-hot and counting penalties are evaluated with reusable workspace. The flow
balances admit the lower workspace-reusing gate upper bound above. On sparse
directed graphs, the flow resources additionally track $|E|$ and the node
degrees, whereas the displayed position model retains explicit
forbidden-transition terms. This qubit--coupler--gate comparison is the
engineering motivation for reporting all three resource notions rather than
only the number of binary variables. Reporting them together exposes the
central lesson: formulation and circuit realization must be assessed jointly.
On dense graphs, problem-qubit and expanded-term counts can favor the position
model, whereas the arithmetic structure of the flow balances yields the
smaller workspace-reusing logical gate-count upper bound under the
implementations analyzed here.

The local functional audit and the hardware experiment separate two further
questions. The audit checks exhaustive scalar energy agreement and, on the
stated deterministic audit set, phase accumulation and work-register cleanup;
it does not empirically validate the optimized asymptotic gate bound. In one Amazon Braket
ProgramSet task, the selected termwise $p=1$ circuits had lower sampled mean
energy than the measured $p=0$ references for both the six-qubit degree-only
Hamiltonian and its reduced 14-qubit flow-augmented counterpart. Because the
provider assigned $p=0$ and $p=1$ to different physical mappings, these
contrasts are descriptive, and their bootstrap intervals quantify conditional
finite-shot uncertainty only. Within each Hamiltonian, the selected and four
sensitivity circuits did share a mapping.

For the degree-only model, $78.05\%$ of the selected-circuit shots were the
invalid disconnected ground state. The selected reduced flow-augmented circuit
had lower sampled mean energy than its separately mapped $p=0$ reference, but
none of the six displayed flow-augmented circuits produced a fully
model-feasible sample in $12{,}000$ shots. Thus low measured energy cannot
repair a connectivity-defective formulation, while a connectivity-correct
Hamiltonian does not by itself guarantee useful feasible sampling on present
hardware. Because the executed circuits used termwise Ising compilation, this
experiment does not measure the resource savings of the reversible
compute--phase--uncompute realization.

The present results concern a feasible, fixed-fleet, homogeneous,
uncapacitated model with nonnegative arc costs. They do not cover capacity,
heterogeneous fleets, physical travel times, time windows, or dynamic routing,
and the analytic logical bounds are not general physical-qubit, depth,
$T$-count, or fault-tolerant estimates. The finite compiled mapping reported
above does not change that limitation. Natural next steps are to compare both
encodings and both circuit realizations under a common hardware-connectivity,
noise, memory, and rotation-precision model and to extend the corrected
construction to capacity, timing, and heterogeneous fleets.

\ifdefined\IEEEBUILD
  \newcommand{\bibtitle}[1]{``#1,''}
\else
  \newcommand{\bibtitle}[1]{#1,}
\fi

\ifdefined\IEEEBUILD
\EOD
\fi

\end{document}